\documentclass[reqno,11pt]{amsart}
\usepackage[foot]{amsaddr}
\usepackage{amssymb}
\usepackage{caption}
\usepackage{graphicx}
\usepackage{subcaption}
\usepackage{threeparttable}
\usepackage{mathrsfs}
\usepackage{float}
\usepackage{dsfont}
\usepackage{amsmath}
\usepackage{diagbox}
\usepackage{booktabs}
\usepackage{array}
\usepackage{multirow}
\usepackage{lineno}
\usepackage{epsfig}
\usepackage{makecell}
\usepackage{changepage}
\usepackage[figuresright]{rotating}
\usepackage[colorlinks=true]{hyperref}
\usepackage{setspace}
\hypersetup{urlcolor=red, citecolor=blue}
\usepackage{cite}
\usepackage{enumitem}
\newtheorem{thm}{Theorem}[section]

\newtheorem{lem}[thm]{Lemma}
\newtheorem{prop}[thm]{Proposition}
\theoremstyle{definition}

\newtheorem{rem}{Remark}[section]
\numberwithin{equation}{section}
\newcommand{\xkh}[1]{\left(#1\right)}
\newcommand{\zkh}[1]{\left[#1\right]}
\newcommand{\dkh}[1]{\left\{#1\right\}}
\newcommand{\xiyu}{\leq& }

\newcommand{\deyu}{=& }
\newcommand{\non}{\nonumber\\}

\title[Predator-mediated apparent competition]{Effects and biological consequences of the predator-mediated apparent competition I: ODE models}

\begin{document}

\begin{abstract}
Predator-mediated apparent competition is an indirect negative interaction between two prey species mediated by a shared predator, which can lead to changes in population dynamics, competition outcomes and community structures. This paper is devoted to investigating the effects and biological consequences of the predator-mediated apparent competition based on a two prey species (one is native and the other is invasive) and one predator model with Holling type I and  II functional response functions.  Through the analytical results and case studies alongside numerical simulations, we find that the initial mass of the invasive prey species, capture rates of prey species, and the predator's mortality rate are all important factors determining the success/failure of invasions and the species coexistence/extinction. The global dynamics can be completely classified for the Holling type I functional response function, but can only be partially determined for the Holling type II functional response function. For the Holling type I response function, we find that whether the invasive prey species can successfully invade to promote the predator-mediated apparent competition is entirely determined by the capture rates of prey species. If the Holling type II response function is applied, then the dynamics are more complicated. First, if two prey species have the same ecological characteristics, then the initial mass of the invasive prey species is the key factor determining the success/failure of the invasion and hence the effect of the predator-mediated apparent competition. Whereas if two prey species have different ecological characteristics, say different capture rates, then the success of the invasion no longer depends on the initial mass of the invasive prey species, but on the capture rates. In all cases, if the invasion succeeds, then the predator-mediated apparent competition's effectiveness essentially depends on the predator's mortality rate. Precisely we show that the native prey species will die out  (resp. persist)  if the predator has a low (resp. moderate) mortality rate, while the predator will go extinct if it has a large mortality rate.  Our study reveals that predator-mediated apparent competition is a complicated ecological process, and its effects and biological consequences depend upon many possible factors.
\end{abstract}

\renewcommand{\thefootnote}{\fnsymbol{footnote}}
\author[Y. Lou, W. Tao and Z.-A. Wang]{Yuan Lou$^{\dagger}$, Weirun Tao$^{\ddagger}$, Zhi-An Wang$^{\S}$}
\email{\rm yuanlou@sjtu.edu.cn (Y. Lou), taoweiruncn@163.com (W. Tao), mawza@polyu.edu.hk (Z.-A. Wang)}
\subjclass[2020]{34D05, 34D23, 92-10, 92D25}
\keywords{Apparent competition, invasion, functional response function, global stability, coexistence and extinction}

\footnotetext[2]{School of Mathematical Sciences, Shanghai Jiao Tong University, Shanghai 200240, China}
\footnotetext[3]{School of Mathematics, Southeast University, Nanjing 211189, China}
\footnotetext[4]{Department of Applied Mathematics, The Hong Kong Polytechnic University, Hung Hom, Kowloon, Hong Kong}
	
\maketitle

\section{Introduction}
Predation is a primary determinant of the structure and function of ecological systems for maintaining biological diversity and balance (cf. \cite{HP1997AN0,S2007E0}). This sounds like a paradoxical statement, as predators kill and consume prey, therefore seeming to cause death, not life. Indeed by doing so, predators may keep other species (like damaging pests) in check and ensure that a multitude of species occupying a variety of environmental niches can survive and thrive. For instance, without the regulation of predators, prey populations may reproduce beyond the carrying capacity of their environments, decimating the populations of smaller animals, plants, and coral reefs (see \cite{url1}). As these species decline, additional organisms that rely on their presence will also decline, resulting in a domino effect that can ultimately push populations and habitats beyond the threshold of recovery. Predators can impact the ecosystem in enormously different ways, and hence gaining a comprehensive understanding of predators' role in ecosystems is a daunting task. Nevertheless, theoretical models alongside analysis can play important parts in interpreting observed patterns/phenomena and making qualitative predictions, and in particular could pinpoint which processes, interactions, or parameter values are responsible for observed behaviors. Competition occurs at the same trophic level, while predation happens between different trophic levels. Though competition and predation can be intertwined directly or indirectly, these two ecological processes are often investigated separately in the existing research.

For the modeling of direct interspecific competition, the growth rate of each species population is described by a first-order differential equation
$$
\frac{dN_i}{dt}=F_i\left(N_1, N_2, \ldots, N_i, \ldots\right).
$$
The species $i$ and $j$ are competing if $\frac{\partial F_i}{\partial N_j}, \frac{\partial F_j}{\partial N_i}<0$ at equilibrium (cf. \cite{M20010}). Indirect interactions between two organisms are mediated or transmitted by a third one. In particular, there is a special indirect negative interaction, called ``apparent competition" (cf. \cite{H1977TPB0,HB2017AREES0}), that happens between victim species mediated through the action of one or more species of shared natural enemies (e.g., predators, herbivores, omnivores, parasitoids, and pathogens). The apparent competition is usually denoted by $(-,-)$, which means a reciprocal negative interaction between each pair of victim species in the presence of a shared natural enemy. Moreover, there are also other types of enemy-mediated indirect interactions, including apparent mutualism $(+,+)$, apparent predation $(+,-)$, apparent commensalism $(+,0)$ and apparent amensalism $(-,0)$ (cf. \cite{CMT2014PMS0,CB2000O0,HB2017AREES0}  and references therein).

In the predator-prey system with one predator and one prey, the specialist predator cannot generally take the prey to extinction as the predators usually starve to death before they can find the last prey. However, if fueled by a secondary prey (called an invasive prey), the predator may take the native prey species to a lower level. This process is called  {the} predator-mediated apparent competition introduced by Holt \cite{H1977TPB0} where a species indirectly and negatively affects another species that shares the same predator by influencing predator abundance of biomass. It has long been recognized as a widespread phenomenon observed in many ecological communities (cf. \cite{CB2000O0,DHHM2010AC0}). In the experiment of \cite{KHE1994O0}, releases of economically unimportant Willamette mites alone, or releases of predatory mites alone, failed to reduce populations of the damaging Pacific spider mite. However, when both herbivorous Willamette mites and predatory mites were released together, populations of Pacific mites were reduced. In \cite{SKB2018E0}, apparent competition between krill and copepods mediated by capelin in the Barents Sea (see a schematic representation in Fig. 1) was employed to advocate that a krill invasion could affect copepod biomass negatively and result in the decrease of copepod biomass. This process involves both bottom-up and top-down effects, where the bottom-up effect influences communities from lower to higher trophic levels of the food web, and the top-down effect is vice versa. However, apparent competition may be difficult to detect or measure due to its indirect nature and the potential for concurrent exploitative competition or other community effects \cite{SKB2018E0}.
\begin{figure}[!ht] \centering
    \includegraphics[width=0.6\textwidth]{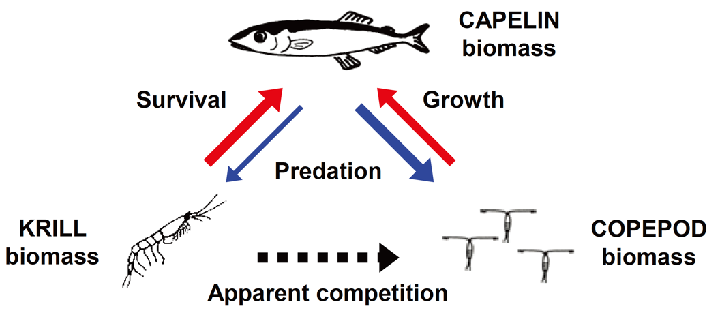}
    \caption{\small Apparent competition between krill and copepods mediated by capelin in the Barents sea. The arrow width is approximately proportional to the strength of the effect size. Bottom-up effects are shown in red, and top-down in blue. (cf. \cite[Fig.1]{SKB2018E0})
    }
\end{figure}

It was pointed out in \cite{HB2017AREES0} that the idea that species can somehow achieve apparent competition by sharing a predator has a venerable history in ecology (cf. \cite{W1957N0} and \cite[pp. 94-95]{L19250}). The mathematical model describing predator-mediated apparent competition was first introduced by Holt \cite{H1977TPB0}, and can be encompassed in the following general model for a single predator species feeding on multiple prey (see also \cite{HB2017AREES0}):
\begin{align}\label{eq1.1}
    \begin{cases}
        \medskip
        \frac{d u_i}{d t}=F_i(\vec{u},w)=u_i\left[g_i\left(u_i\right)-f_i(\vec{u}) w\right],\\
        \frac{d w}{d t}=G(\vec{u},w)=w F(\vec{u}),&\\
    \end{cases}
\end{align}
where $w$ and $u_i$ are densities of the predator and prey species $i$, and the arrow over $u$ denotes a vector of prey abundances, $F_i$ is the total growth rate of prey species $i$ and $G$ is the growth rate of the predator. In the first equation of \eqref{eq1.1}, $g_i(u_i)$ is the inherent per capita growth rate of the prey $i$ in the absence of the predator, $f_i(\vec{u})$ is the functional response function of the predator to prey species $i$ and the quantity $f_i(\vec{u}) w$ is the per capita rate of mortality from predation experienced by prey species $i$. The right-hand side of the second equation of \eqref{eq1.1} states that the per capita growth rate $F(\vec{u})$ of the predator depends on prey availability.  Focusing on {the} predator-mediated apparent competition (i.e. indirect interaction), it is assumed in \eqref{eq1.1} that direct interspecific competition among prey species is negligible.

Though the importance of {the} predator-mediated apparent competition was extensively discussed in various biological literature (see  \cite{CB2000O0,SKB2018E0,KHE1994O0,DHHM2010AC0} and references therein) since the pioneering work \cite{H1977TPB0}, mathematical works are much less compared to those for the predator-prey or competition systems (e.g. see \cite{CC20030,C2014DCDS0,KW2013JMB0,MBN20130,N20110,RC2015DCDS0,STA2003AN0,WZZ2016JMB0} and references therein). Based on \eqref{eq1.1}, in this paper, we consider a specified predator-mediated apparent competition model with two prey species and one shared predator:
\begin{equation}\label{model}
    \begin{cases}
        \medskip u_{t}= u\left(1-u / K_{1}\right)-w f_{1}(u),  \quad &{t>0,}\\
        \medskip v_{t}= v\left(1-v / K_{2}\right)-w f_{2}(v),  \quad &{t>0,}\\
         \medskip w_{t}=w\xkh{\beta_{1} f_{1}(u)+\beta_{2} f_{2}(v)-\theta},  \quad &{t>0,}\\
        {(u,v,w)(0)=(u_0,v_0,w_0)}&
    \end{cases}
\end{equation}
where $u(t), v(t)$ and $w(t)$ represent the densities of native prey species, the invasive prey species, and shared predator species at time $t$, respectively. {The initial data $u_0,v_0,w_0$ are assumed to be positive.} The function $f_i(i=1,2)$ and parameters have the following biological interpretations:
\begin{spacing}{1.1}
    \begin{itemize}
        \item $f_i$, $i=1,2$,  - functional response functions;
        \item $K_i$, $i=1,2$,  - carrying capacities for the prey species;
        \item $\beta_i$, $i=1,2$,  - trophic efficiency (conversion rates);
        \item $\theta$ - mortality rate of the predator.
    \end{itemize}
\end{spacing}

All the parameters shown above are positive. For definiteness, we consider two types of the functional response functions:
\begin{align}
    \medskip
    &f_i(s)=\alpha_i s,\qquad\quad\ i=1,2,\ (\text{Holling type I}),\label{eqh1}\\
    &f_i(s)=\frac {\gamma_i s}{1+\gamma_i h_i s},~\ i=1,2, \ (\text{Holling type II}),\label{eqh2}
\end{align}
where $\alpha_i$ and $\gamma_i$, $i=1,2$, denote the capture rates (i.e. the rates at which prey species are captured), and $h_i>0$, $i=1,2$, represents the handling time.

If the functional response functions are Holling type I (i.e. Lotka-Volterra type) and the direct competition between the two prey species is considered in \eqref{model},  the predator-mediated coexistence and/or extinction have been studied in the literature \cite{V1978AN0,PVM2017SJADS0,MIMURA1986129,H1981MB0,C1978AN0,A1999E0}. However, these works did not consider the case of Holling type II functional response function.  The main determinants of the predator-mediated apparent competition and quantitative effects of the predator-mediated apparent competition were not investigated either. In this paper, we shall consider these questions not studied in the literature.

Since the predator-mediated apparent competition involves an invasion of the secondary prey species (called invasive prey species in the sequel) which is also a food supply to the shared predator, the invasion may not be successful and consequently, the predator-mediated apparent competition will not take effect. Therefore the first aim of this paper is to investigate
\begin{itemize}[leftmargin=20mm]
\item[\textbf{A1}.] Under what conditions, the invasive prey species can successfully invade to promote the predator-mediated apparent competition?
\end{itemize}
If the invasive prey species invade successfully and supply additional food to the predator, then the native prey species will be under more intensive predation pressure, possibly resulting in a population decrease or even extinction.  Hence the second aim of this paper is to address
\begin{itemize}[leftmargin=20mm]
    \item[\textbf{A2}.] Whether the predator-mediate apparent competition could reduce the biomass of the native prey species or even drive the native species to go extinct? If so, what conditions are required, and which processes are the main determinants?
\end{itemize}

By global and local stability analysis alongside numerical simulations, we find that the answers to the above questions depend upon the types of the functional response function. For the Holling type I response function, we find that whether the invasive prey species can successfully invade to launch the predator-mediated apparent competition is entirely determined by the capture rates of native and invasive prey species as well as the predator's mortality rate. Whereas for the Holling type II response function, the dynamics are more complicated. First, if two prey species have the same ecological characteristics (i.e. symmetric apparent competition), then the initial mass of the invasive prey species is the key factor determining the success of the invasion and effectiveness of the predator-mediated apparent competition. However, if two prey species have different ecological characteristics (i.e. asymmetric apparent competition), say different capture rates, then the success of the invasion does not rely on the initial mass of the invasive prey species, but on the capture rates of both prey species. In all cases, given that the invasion succeeds, the effectiveness of the predator-mediated apparent competition essentially depends on the predator’s mortality rate. Precisely, the native prey species will die out (resp. persist) if the predator has a low (resp. moderate) mortality rate, while the predator will go extinct if it has a large mortality rate. Moreover, we find that the global dynamics of \eqref{model} with Holling type I functional response function can be completely determined and there is no non-constant solution (like periodic solution). In contrast, \eqref{model} with Holling type II functional response function can generate complex dynamics including periodic pattern (limit cycle) and bistability. This implies that  Holling type I functional response function is unsuitable to explain the temporal inhomogeneity.


The rest of this paper is organized as follows. In Sect. \ref{sec2}, we state our main mathematical results on the global stability of {the} system \eqref{model} with \eqref{eqh1} and \eqref{eqh2}, and the relevant proofs are given in Sect. \ref{sec3}. In Sect. \ref{sec4}, we focus on the case of Holling type II functional response function and conduct case studies to pinpoint the main factors determining the effects and biological consequence of the predator-mediated apparent competition. In Section \ref{sec5}, we {summarize} our main findings and discuss several open questions.

\section{Global stability results}\label{sec2}
This section outlines our primary mathematical findings. We first introduce some notations used throughout the paper for clarity and brevity and then proceed to state the main results. {Let}
\begin{align*}
    \begin{array}{ll}
        \medskip
        L_i:=\beta_if_i(K_i),\qquad &\lambda_i:=\frac1{\gamma_ih_i},\quad i=1,2,\\
        L :=L_1+L_2,\quad
        &\theta_0:=\max
        \dkh{
        (1-\frac{\alpha_1}{\alpha_2})L_1,
        (1-\frac{\alpha_2}{\alpha_1})L_2
        }.
    \end{array}
\end{align*}
{We denote} the equilibrium of \eqref{model} by $E_s=\xkh{u_s,v_s,w_s}${, which includes extinction equilibria, predator-free} equilibria, {semi-coexistence equilibria and coexistence equilibria} listed in Table \ref{table1}, where the {coexistence equilibrium} $E_*=(u_*,v_*,w_*)$ is obtained by solving \eqref{model} for {$u,v,w>0$}. To differentiate {coexistence equilibria} {for} different functional response functions, we utilize the notation
\begin{align} \nonumber
    E_*=
    \begin{cases}
        \medskip P_*,\qquad &\text{if }\eqref{eqh1}\text{ holds},\\
        Q_*,&\text{if }\eqref{eqh2}\text{ holds}.
    \end{cases}
\end{align}
Moreover, in the case of Holling type I functional response function \eqref{eqh1}, the {coexistence equilibrium} $P_*$ is uniquely given by
\begin{align}\nonumber
    P_*=\xkh{\frac{K_{1} \zkh{(\alpha_2-\alpha_1)L_2+\alpha_1\theta}}{\alpha_1L_1+\alpha_2L_2}, \frac{K_{2} \zkh{(\alpha_1-\alpha_2)L_1+\alpha_2\theta}}{\alpha_1L_1+\alpha_2L_2}, \frac{L-\theta}{\alpha_1L_1+\alpha_2L_2}},
\end{align}
while in the case of Holling type II functional response function \eqref{eqh2}, the {coexistence equilibrium} $Q_*$ may not exist, be unique, or exist but not be unique (see Remark \ref{rem2.1}).

\begin{rem}\label{rem2.1}
    For {the} system \eqref{model} with Holling type II functional response function \eqref{eqh2}, it is difficult to {find} the necessary and sufficient conditions for the existence of $Q_*$ for general system parameters. Note that $0<\theta<L$ is a necessary but not sufficient condition for the existence of $Q_*$. Indeed, the necessity is apparent since it is easy to see that $u_*<K_1$, $v_*<K_2$, and thus
    \begin{align}\nonumber
        \theta=\beta_{1} f_{1}(u_*)+\beta_{2} f_{2}(v_*)<\beta_{1} f_{1}(K_1)+\beta_{2} f_{2}(K_2)=L,
    \end{align}
    where we have used the fact that $f_i(s)$, $i=1,2$, strictly increases with respect to $s>0$. However, if
    \begin{align*}
        \theta=\frac35,\ K_1=2,\ K_2=3,
        \quad\text{and}\quad
        \beta_i=\gamma_i=h_i=1,\ i=1,2,
    \end{align*}
    then {the} system \eqref{model} with \eqref{eqh2} has no any {coexistence equilibrium though} ${0<}\theta<L=\frac{17}{12}$.
\end{rem}

{
\renewcommand{\arraystretch}{1}
\begin{table}[!ht]
    \small
    \caption{Equilibria of {the} system \eqref{model} with \eqref{eqh1} or \eqref{eqh2}.}
    \begin{tabular}
        {c|c|c|c}
    \hline
    \multicolumn{2}{c|}{\makecell{Type of \\equilibria}}
    & Expression of equilibria
    & \makecell{Necessary and\\ sufficient condition}
    \rule{0pt}{22pt}  \\[2ex]
    \hline
    \multicolumn{2}{c|}{\makecell{{Extinction} \\equilibria}}
    &
    $
    E_0=(0,0,0)
    $
    &$\theta>0$
    \rule{0pt}{22pt} \\[2.5ex]
     \hline
     \multicolumn{2}{c|}{\makecell{{Predator-free} \\equilibria}}
    &
    $
    E_u=(K_1,0,0),\ E_v=(0,K_2,0),\ E_{uv}=(K_1,K_2,0)
    $
    &$\theta>0$
    \rule{0pt}{22pt} \\[2.5ex]
     \hline
    {\multirow{4}{*}{
      \rule{0pt}{50pt}\makecell{{Semi-}\\{coexistence}\\equilibria}}}
    &{\multirow{2}{*}{   \rule{0pt}{20pt}\eqref{eqh1}}}
    &     $P_1=\xkh{u_{P_1},0,w_{P_1}}=\left(\frac\theta{\alpha_1\beta_1}, 0,\frac{L_1-\theta}{\alpha_1L_1}\right)$
    &         $0<\theta<L_1$
   \rule{0pt}{17pt}\\[1.5ex]
    & & $  P_2=\xkh{0,v_{P_2},w_{P_2}}=\left(0,\frac\theta{\alpha_2\beta_2},\frac{L_2-\theta}{\alpha_2L_2}\right)$
    &          $0<\theta<L_2$
    \\[1.5ex]
    \cline{2-4}
    &{\multirow{2}{*}{   \rule{0pt}{20pt}\eqref{eqh2}}}
    &  $Q_1=\xkh{u_{Q_1},0,w_{Q_1}}=\xkh{ \frac{\theta  }{(\beta_1 -h_1 \theta)\gamma_1},  0,
    \frac{\beta_1 \xkh{L_1 -\theta}}{ \gamma_1f_1(K_1)(\beta_1 - h_1 \theta )^2}}$
    &         $0<\theta<L_1$
 \rule{0pt}{17pt}\\[1.5ex]
    & & $Q_2=\xkh{0,v_{Q_2},w_{Q_2}}=\xkh{0, \frac{\theta }{ (\beta_2 -  h_2 \theta)\gamma_2},
    \frac{\beta_2  \xkh{L_2-\theta}}{ \gamma_2 f_2(K_2) (\beta_2 -  h_2 \theta)^2}}$
    &          $0<\theta<L_2$
  \\[1.5ex]
    \hline
    {\multirow{2}{*}{
      \rule{0pt}{22pt}\makecell{{Coexistence}\\equilibria}}}
    & \eqref{eqh1}
    & $P_*$
    &   $\theta_0<\theta<L$
        \rule{0pt}{17pt} \\[1ex]
        \cline{2-4}
    & \eqref{eqh2}
    & $Q_*$
    &          Unclear (see Remark \ref{rem2.1})
      \rule{0pt}{17pt} \\[1ex]
    \hline
    \end{tabular}
    \label{table1}
\end{table}
}

Clearly we have $L_1,L_2,L>0$, $0\leq \theta_0<L$ and $\theta_0=0$ if and only if $\alpha_1=\alpha_2$. For the global stability of equilibria of systems \eqref{model}, it is easy to find that the equilibria $E_0$, $E_u $, $E_v $ are saddles for $\theta>0$, and $E_{uv}$ is also a saddle for $\theta\in(0,L)$ (see Lemma \ref{lem4.1}). Therefore, we will focus on analyzing the global stability of the equilibrium $E_{uv}$ for $\theta\geq L$, and the {semi-coexistence}/{coexistence equilibria} for $\theta<L$. Now we can state our main results.

\begin{thm}[Global stability for Holling type I]\label{thm2.1}
Let $f_1(u)$ and $f_2(v)$ be given by \eqref{eqh1}. Then the following global stability results hold for \eqref{model}.
    \begin{itemize}[leftmargin=10mm]
        \item[(i)] If $\alpha_1<\alpha_2$ (resp. $\alpha_1>\alpha_2$) and $\theta\in(0,\theta_0]$, then the {semi-coexistence} equilibrium $P_1$ (resp. $P_2$) is globally asymptotically stable.
        \item[(ii)] If $\theta\in(\theta_0,L )$, then the unique {coexistence equilibrium} $P_*=\xkh{u_*,v_*,w_*}$ of \eqref{model} is globally asymptotically stable.
        \item[(iii)]  If $\theta\geq L$, then the equilibrium $E_{uv}$ is globally asymptotically stable.
    \end{itemize}
\end{thm}

\begin{thm}[Global stability for Holling type II]\label{thm2.2}
    Let $f_1(u)$ and $f_2(v)$ be given by \eqref{eqh2}. Then the following global stability results hold for \eqref{model}.
        \begin{itemize}[leftmargin=10mm]
            \item[(i)] Let $\theta\in(0,L_1)$. Then the {semi-coexistence} equilibrium $Q_1$ is globally asymptotically stable if
            \begin{align}\label{eq2.1}
                (K_1,K_2)\in\Lambda_1:=\dkh{(K_1,K_2)\ \bigg|\ K_1\leq\lambda_1+u_{Q_1},\ \frac{K_2}{f_2(K_2)}\leq w_{Q_1}},
            \end{align}
            where ``='' in $\frac{K_2}{f_2(K_2)}\leq w_{Q_1}$ holds only in the case of $v_0\leq K_2$.
            \item[(ii)] Let $\theta\in(0,L_2)$. Then the {semi-coexistence} equilibrium $Q_2$ is globally asymptotically stable if
            \begin{align}\label{eq2.2}
                (K_1,K_2)\in\Lambda_2:=\dkh{(K_1,K_2)\ \bigg|\ K_2\leq\lambda_2+v_{Q_2},\ \frac{K_1}{f_1(K_1)}\leq w_{Q_2}},
            \end{align}
            where ``='' in $\frac{K_1}{f_1(K_1)}\leq w_{Q_2}$ holds only in the case of $u_0\leq K_1$.
            \item[(iii)] Let $\theta\in(0,L)$ and coexistence equilibrium $Q_*=\xkh{u_*,v_*,w_*}$ exist. Then $Q_*$ is globally asymptotically stable if
            \begin{align}\label{eq2.3}
                (K_1,K_2)\in\Lambda_*:=\dkh{(K_1,K_2)\ \bigg|\ K_1\leq \lambda_1+u_*,\ K_2\leq \lambda_2+v_*}.
            \end{align}
            \item[(iv)] Let $\theta\geq L$. Then the equilibrium $E_{uv}$ is globally asymptotically stable.
        \end{itemize}
\end{thm}

\begin{rem}\label{rem2.2}
We note that the  sets $\Lambda_1$, $\Lambda_2$ and ${\Lambda_*}$ given in \eqref{eq2.1}-\eqref{eq2.3} are mutually disjoint. See Appendix A for the detailed {proof}.
\end{rem}

\begin{rem}
    In view of Theorem \ref{thm2.1}, the global stability of {the} system \eqref{model} with Holling type I functional response function \eqref{eqh1} can be completely classified, as summarized in Table \ref{table2}. However, for Holling type II functional response function \eqref{eqh2}, there are some gaps (see Table \ref{table3}) for $0<\theta<L$ left for the global stability.
\end{rem}
{
\renewcommand{\arraystretch}{1.5}
\begin{table}[!ht]
    \caption{\small Global stability of equilibria of {the} system \eqref{model} with \eqref{eqh1}.}
    \begin{tabular}
        {m{3.5cm}<{\centering}|m{3.5cm}<{\centering}|m{3cm}<{\centering}|m{3cm}<{\centering}}
        \hline
        &$\theta\in(0,\theta_0]$
        &$\theta\in(\theta_0,L)$
        &$\theta\in[L,\infty)$\\
        \hline
        $\alpha_1>\alpha_2$& $P_2$ is {GAS}
        & {\multirow{3}{*}{$P_*$ is {GAS}}}
        & {\multirow{3}{*}{$E_{uv}$ is {GAS}}}\\
        \cline{1-2}
        $\alpha_1<\alpha_2$& $P_1$ is {GAS} & &\\
        \cline{1-2}
        \multicolumn{2}{c|}{$\alpha_1=\alpha_2\ (\Longleftrightarrow \theta_0=0)$}
        &
        &
        \\
        \hline
    \end{tabular}
    \label{table2}
    \begin{adjustwidth}{11mm}{17mm}
        \begin{tablenotes}
            \footnotesize
            \item {Note: {Here the notations ``{GAS}" and ``$\Longleftrightarrow$'' denote “globally asymptotically stable” and ``if and only if'', respectively.}}
        \end{tablenotes}
    \end{adjustwidth}
\end{table}
}

{
\renewcommand{\arraystretch}{1.5}
\begin{table}[!ht]
    \caption{ Global stability of equilibria of {the} system \eqref{model} with \eqref{eqh2}.}
    \begin{tabular}
        {m{4cm}<{\centering}|m{3cm}<{\centering}|m{3cm}<{\centering}|m{3cm}<{\centering}}
        \hline
        $i\in\dkh{1,2}$
        &$\theta\in(0,L_i)$
        &$\theta\in[L_i,L)$
        &$\theta\in[L,\infty)$\\
        \hline
        $(K_1,K_2)\in\Lambda_i$& $Q_i$ is {GAS}
        & Unclear
        & {\multirow{3}{*}{$E_{uv}$ is {GAS}}}\\
        \cline{1-3}
        $(K_1,K_2)\in\Lambda_*$&
        \multicolumn{2}{c|}{$Q_*$ is {GAS}}
        &\\
        \cline{1-3}
        $(K_1,K_2)\not\in\Lambda_1\cup\Lambda_2\cup\Lambda_*$
        &
        \multicolumn{2}{c|}{Unclear}
        &
        \\
        \hline
    \end{tabular}
    \label{table3}
    \begin{adjustwidth}{11mm}{17mm}
        \begin{tablenotes}
            \footnotesize
            \item {Note: {Here the notation ``{GAS}" has the same interpretation as in Table \ref{table2}.}}
        \end{tablenotes}
    \end{adjustwidth}
\end{table}
}
The global stability results asserted in Theorem \ref{thm2.1} and Theorem \ref{thm2.2} will be proved by the Lyapunov function method along with LaSalle's invariant principle. The proofs will be postponed to Section 5.

\vskip3mm
\section{{Numerical} simulations {and biological implications}}\label{sec4}

From Table \ref{table2}, we see that the global stability of solutions to \eqref{model} with \eqref{eqh1} has been completely classified and there is no gap left for the global stability of solutions. In contrast, there are some parameter gaps in which the global dynamics of \eqref{model} with \eqref{eqh2} remain unknown (see Table \ref{table3}). In the following, we shall numerically explore the global dynamics of \eqref{model} with \eqref{eqh2} in these gaps.
It is well known that one predator and one prey models with Holling type II functional response function may have stable time-periodic solutions \cite{C1981SIAMJMA0}. Therefore we anticipate that periodic solutions {may} arise from the system \eqref{model} with Holling type II functional response function. Apart from this, we shall {also} investigate the effect of the predator-mediated apparent competition on population dynamics.

The associated Jacobian matrix of {the} system \eqref{model} at an equilibrium $E_s=(u_s, v_s, w_s)$ is
\begin{align*}
    \mathcal{J}(E_s)
    &=
\left(
\begin{array}{ccc}
    \medskip 1-\frac{2 u_s}{K_1}-w_s f_1'(u_s) & 0 & -f_1(u_s) \\
    \medskip 0 & 1-\frac{2 v_s}{K_2}-w_s f_2'(v_s) & -f_2(v_s) \\
    \medskip \beta_1 w_s f_1'(u_s)& \beta_2 w_s f_2'(v_s) & \beta_{1} f_{1}(u_s)+\beta_{2} f_{2}(v_s)-\theta  \\
\end{array}
\right)
\non
&=:
\left(
\begin{array}{ccc}
    \medskip J_{11} & 0 & J_{13} \\
    \medskip 0 & J_{22} & J_{23} \\
 J_{31}& J_{32}&J_{33}  \\
\end{array}
\right).
\end{align*}
We denote the three eigenvalues of $\mathcal{J}(E_s)$ by $\rho_1$, {$\rho_-$ and $\rho_+$}, which are the roots of
\begin{align}\label{eq4.1}
    \rho^3+a_2\rho^2+a_1\rho+a_0=0,
\end{align}
where $a_i=a_i(E_s)$, $i=0,1,2$, are given by
\begin{eqnarray}\nonumber
\left\{
\begin{array}{llll}
     a_0:= J_{11} J_{22} J_{33}-J_{11} J_{23} J_{32}-J_{13} J_{22} J_{31},\\
    a_1:=J_{11} J_{22}+J_{11} J_{33}+J_{22} J_{33}-J_{13} J_{31}-J_{23} J_{32},\\
        a_2:=-(J_{11}+J_{22}+J_{33}).
\end{array}
\right.
\end{eqnarray}
It follows from the Routh-Hurwitz criterion (cf. \cite[Appendix B]{M20020}) that all roots of (\ref{eq4.1}) have negative real parts if and only if
\begin{equation}\nonumber
    a_0,a_1,a_2>0 \quad\text{and}\quad a_1a_2-a_0>0.
\end{equation}

Next, we use the above results to study the stability of all equilibria. First from Theorem \ref{thm2.2} (iv) it follows that $E_{uv}$ is globally asymptotically stable for $\theta\geq L$. The following results can also be easily obtained.

\begin{lem}\label{lem4.1}
    The equilibria $E_0$, $E_u $, $E_v $ are saddles for {any} $\theta>0$. The equilibrium $E_{uv}$ is a saddle for $\theta\in(0,L )$, {while} $E_{uv}$ is globally asymptotically stable for $\theta\geq L$.
\end{lem}
\begin{proof}
    With simple calculations, one can easily find that the eigenvalues of $\mathcal{J}$ at the four equilibria $E_0, E_u, E_v, E_{uv}$ are
    \begin{align}\nonumber
        \left\{
            \begin{array}{lll}
                \medskip \rho_1= -\theta,&\rho_\pm = 1,&\text{if }E_s=E_0,\\
                \medskip \rho_1= L_1-\theta,\qquad&\rho_\pm = \pm1,\qquad\qquad&\text{if }E_s=E_u,\\
                \medskip \rho_1= L_2-\theta,&\rho_\pm = \pm1,\quad&\text{if }E_s=E_v,\\
                \rho_1= L -\theta,\quad&\rho_\pm = -1,\quad&\text{if }E_s=E_{uv},
            \end{array}
        \right.
    \end{align}
    which completes the proof.
\end{proof}

We next investigate the stability of semi-coexistence equilibria $Q_1$, $Q_2$, and coexistence equilibria $Q_*$. It turns out that the stability analysis for these equilibria of \eqref{model} with Holling type II functional response function \eqref{eqh2} is too complicated to find explicit stability/instability conditions. for clarity and definiteness,
we assume that the handling time for the two prey species is the same by simply letting $h_1=h_2=1$ for brevity. By \eqref{eqh2}, it holds that
\begin{align}\label{eq4.2}
    f_i(s)=\frac {s}{\frac1{\gamma_i}+ s}=:\frac {s}{\lambda_i+ s},
    \,\quad s\geq0,\ i=1,2.
\end{align}
{In what follows, we} shall use \eqref{eq4.2} instead of \eqref{eqh2} as the Holling type II functional response function to undertake case studies along with numerical simulations. As illustrated in \cite[Figure 1]{HB2017AREES0}, predator-mediated apparent competition among two victim prey species may be symmetric or asymmetric. Hence we shall distinguish these two scenarios in our subsequent analysis.
\begin{itemize}
    \item \textbf{{Symmetric apparent competition}} (two prey species have the same ecological characteristics): Two prey species have the same ecological characteristics, namely the two prey species are different phenotypes of the same species. In this case, we will consider
 $$K_i=K,\ \beta_i=\beta,\ \gamma_i=\gamma,\ h_i=h,\quad i=1,2,$$   where $K,\beta,\gamma$ and $h$ are positive constants.
     \item \textbf{{Asymmetric apparent competition}} (two prey species have different ecological characteristics): Prey species having different ecological characteristics may be dissimilar in many ways, such as the carrying capacity, trophic efficiency, the rate of being captured by the predator (i.e. capture rate), and so on. In this case, we may assume that the two prey species have different values for one parameter and the same values for other parameters.
     \end{itemize}
\subsection{\textbf{{Symmetric apparent competition}}} \label{sbusec4.1}
For definiteness and simplicity of computations, without loss of generality, we take
    \begin{align}
        K_1=K_2=3
        \quad\text{and}\quad
        \beta_1=\beta_2=\lambda_1=\lambda_2=1.\label{eq4.3}
    \end{align}
We deduce from  \eqref{eq4.3} that $L_1=L_2=\frac34$ and $L=\frac32$. In addition to the equilibria $E_0$, $E_u$, $E_v$ and $E_{uv}$ of \eqref{model} which exist for any $\theta>0$, there are two {semi-coexistence} equilibria
\begin{align}\label{eq4.5}
    \begin{cases}
        \medskip
        Q_1=\xkh{\frac{ \theta }{1-\theta},0,\frac{3-4\theta }{3 (1-\theta )^2}},\\
        Q_2=\xkh{0,\frac{ \theta }{1-\theta},\frac{3-4\theta }{3 (1-\theta )^2}},
    \end{cases}
    \quad \ \text{if} \ \theta\in\xkh{0,\frac34}.
\end{align}
With tedious but elementary calculations, one can find that there is no {coexistence equilibrium} if $\theta\geq \frac32$, a unique {coexistence equilibria} ${Q_*^0}$ {exists} if $\theta\in(0,\frac23]\cup[1,\frac32)$ and three {coexistence equilibria} $Q_*^i$ {(${i=0,1,2}$) exist} if $\theta\in(\frac23,1)$, where
\begin{align}\label{eq4.6}
    \begin{cases}
        \medskip
        {Q_*^0}:=\xkh{\frac\theta{2-\theta},\frac\theta{2-\theta},\frac{4 (3-2 \theta )}{3 (2-\theta )^2}},\\
        \medskip
        {Q_*^1}:=\xkh{1+2 \sqrt{\frac{1-\theta }{2-\theta }},1-2 \sqrt{\frac{1-\theta }{2-\theta }},\frac4{3(2-\theta)}},\\
        {Q_*^2}:=\xkh{1-2 \sqrt{\frac{1-\theta }{2-\theta }},1+2 \sqrt{\frac{1-\theta }{2-\theta }},\frac4{3(2-\theta)}}.
    \end{cases}
\end{align}

{
\begin{rem}\label{rem_Thm2.2_apply1}
    In addition to the global stability result for $E_{uv}$ stated in Lemma \ref{lem4.1}, we can also apply Theorem \ref{thm2.2} (iii) to see that ${Q_*^0}$ is globally asymptotically stable for $\theta\in[\frac43,\frac32)$ since $u_*=v_*=\frac\theta{2-\theta}\geq2=K_i-\lambda_i$ ($i=1,2$).
\end{rem}
}

In view of Lemma \ref{lem4.1} and Remark \ref{rem_Thm2.2_apply1}, it remains to consider the stabilities of {semi-coexistence} and {coexistence equilibria} for $\theta\in(0,\frac32)$. We begin with the local stability {of} the {semi-coexistence} equilibria $Q_1$ and $Q_2$ for $\theta\in(0,\frac34)$.

\begin{lem}
    Let \eqref{eq4.3} hold and $\theta\in(0,\frac34)$. {Then $Q_i$ ($i=1,2$) has the following properties.}
    \begin{itemize}[leftmargin=10mm]
        \item If $\theta\in\{\frac12,\frac23\}$, then $Q_i$ ($i=1,2$) is marginally stable, where $\rho_1=-\frac13, \rho_\pm=\pm\frac i{\sqrt6}$ if $\theta=\frac12$, and $\rho_1=0, \rho_\pm=\frac{-2\pm\sqrt{2} i}{9}$ if $\theta=\frac23$.
        \item If $\theta\in(0,\frac12)$, then $Q_i$ {is} a saddle-focus, where $\rho_1<0$, and $\rho_\pm$ are a pair of complex-conjugate eigenvalues with ${\rm Re}(\rho_\pm)>0$ and $\rm Im (\rho_\pm)\neq0$.
        \medskip
        \item If $\theta\in(\frac12,\frac23)$, then $Q_i$ {is} a stable focus-node, where $\rho_1<0$, and $\rho_\pm$ are a pair of complex-conjugate eigenvalues with ${\rm Re}(\rho_\pm)<0$ and $\rm Im (\rho_\pm)\neq0$.
        \medskip
        \item If $\theta\in(\frac23,\theta_1)$, then $Q_i$ {is} a saddle-focus, where $\rho_1>0$, and $\rho_\pm$ are a pair of complex-conjugate eigenvalues with ${\rm Re}(\rho_\pm)<0$ and $\rm Im (\rho_\pm)\neq0$.
        \medskip
        \item If $\theta\in[\theta_1,\frac34)$, then $Q_i$ is a saddle with $\rho_1>0$ and $\rho_\pm <0$.
    \end{itemize}
    Here, $\theta_1\approx0.6793$ is the unique real root of the equation $16 \theta ^3-37 \theta ^2+31 \theta -9=0$ for $\theta\in(0,\frac34)$.
\end{lem}
\begin{proof}
    {We omit the proofs for brevity as they are elementary.}
\end{proof}

We next give the local stability of the {coexistence equilibria}.
\begin{lem}
    Let \eqref{eq4.3} hold and $\theta\in(0,\frac32)$. {Then ${Q_*^0}$ has the following properties.}
    \begin{itemize}[leftmargin=10mm]
        \item If $\theta=1$, then ${Q_*^0}$ is marginally stable with $\rho_1=0$ and $\rho_\pm=\pm\frac i{\sqrt3}$.
        \item If $\theta\in(0,1)$, then $\rho_1>0$, and $\rho_\pm$ are a pair of complex-conjugate eigenvalues with $\rm Re(\rho_\pm)>0$ and $\rm Im (\rho_\pm)\neq0$. Therefore, ${Q_*^0}$ is an unstable focus-node.
        \medskip
        \item If $\theta\in(1,\frac34)$, then $\rho_1<0$, and $\rho_\pm$ are a pair of complex-conjugate eigenvalues with ${\rm Re}(\rho_\pm)<0$ and $\rm Im (\rho_\pm)\neq0$. As a result, ${Q_*^0}$ is a stable focus-node.
        \medskip
        \item If $\theta\in[\frac34,\frac23)$, then ${Q_*^0}$ is globally asymptotically stable.
    \end{itemize}
\end{lem}
\begin{proof}
    The proofs of the first two conclusions are omitted for brevity since they are standard and elementary. The third conclusion is a direct consequence of Theorem \ref{thm2.2} (iii), see Remark \ref{rem_Thm2.2_apply1}.
\end{proof}
{With some tedious calculations, we also obtain the following result.}
\begin{lem}\label{lem4.4}
    Let \eqref{eq4.3} hold and $\theta\in(\frac23,1)$. Then $\rho_1<0$, and $\rho_\pm$ are a pair of complex-conjugate eigenvalues with ${\rm Re}(\rho_\pm)<0$ and $\rm Im (\rho_\pm)\neq0$. Hence ${Q_*^1}$ and ${Q_*^2}$ are stable focus-nodes.
\end{lem}

{\renewcommand{\arraystretch}{1.5}
\begin{table}[!h]\small
    \caption{The stability of equilibria of system \eqref{model} with \eqref{eq4.3}.}
    \begin{tabular}
        {c|c|c|c|c|c|c|c|c|c|c|c}
        \hline
        \diagbox{\footnotesize Equilibria}{$\theta$}
        & $(0,\frac12)$
        & $\frac12$
        & $(\frac12,\frac23)$
        & $\frac23$
        & $(\frac23,\theta_1)$
        & $[\theta_1,\frac34)$
        & $[\frac34,1)$
        & $1$
        & $(1,\frac43)$
        & $[\frac43,\frac32)$
        & $[\frac32,\infty)$\\
        \hline
        $E_0, E_u, E_v $
        & \multicolumn{11}{l}{\makecell[c]{Saddle}}\\
        \hline
        $E_{uv}$
        & \multicolumn{10}{l|}{\makecell[c]{Saddle}}
        &{GAS}\\
        \hline
        $Q_1$, $Q_2$
        &SF
        &MS
        &S-FN
        &MS
        &SF
        &Saddle
        & \multicolumn{5}{l}{\makecell[c]{/}}\\
        \hline
        ${Q_*^0}$
        &\multicolumn{7}{c|}{U-FN}
        &MS
        &S-FN
        &{GAS}
        &/\\
        \hline
        ${Q_*^1},{Q_*^2}$
        & \multicolumn{4}{l|}{\makecell[c]{/}}
        & \multicolumn{3}{l|}{\makecell[c]{S-FN }}
        & \multicolumn{4}{l}{\makecell[c]{/}} \\
        \hline
    \end{tabular}
    \label{table4}
    \begin{adjustwidth}{0cm}{6mm}
        \begin{tablenotes}
            \footnotesize
            \item {Note: The abbreviations ``MS", ``SF", ``S-FN", and ``U-FN" stand for ``marginally stable", ``saddle-focus", ``stable focus node", and ``unstable focus node", respectively. {The notation ``{GAS}" has the same interpretation as in Table \ref{table2}}. The notation ``/" denotes “equilibria do not exist” and  $\theta_1\approx0.6793$ is given in Lemma \ref{lem4.1}}.
        \end{tablenotes}
    \end{adjustwidth}
\end{table}
}

\begin{figure}[!ht] \centering
    \includegraphics
    [width=1\textwidth,trim=0 0 0 0,clip]{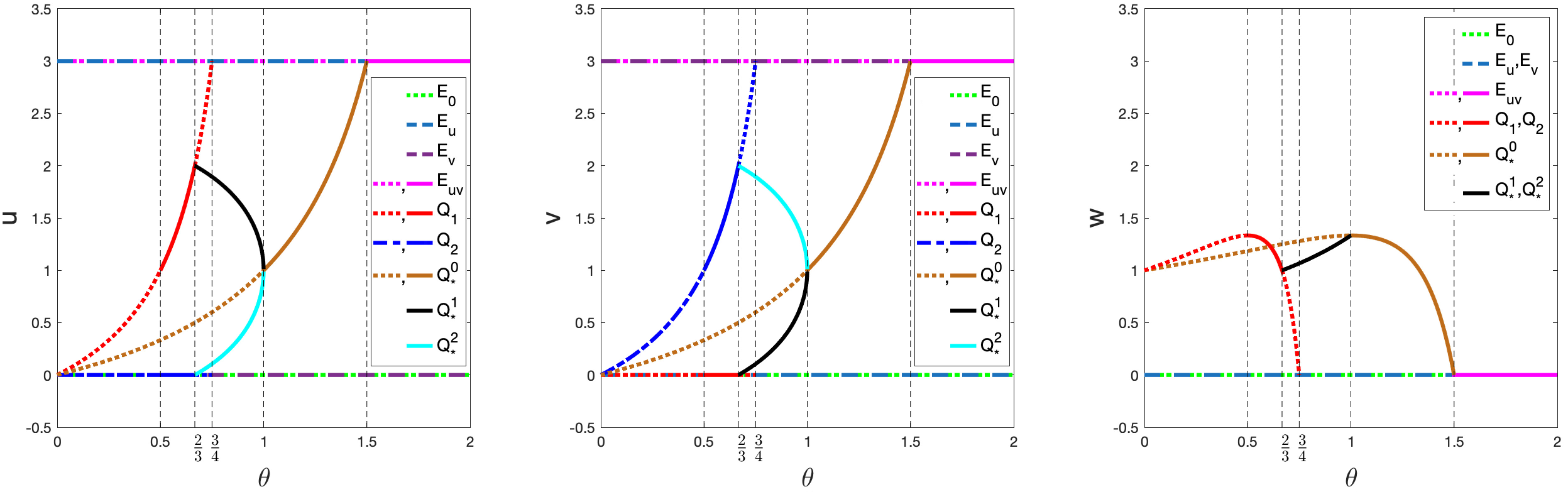}
    \caption{\small
    {Bifurcation diagrams of system \eqref{model} with \eqref{eq4.3} versus $\theta$. The solid curves denote linearly stable equilibria, and other types of curves represent unstable equilibria.}
    }
    \label{fig2}
\end{figure}

\begin{figure}[!ht] \centering
    \includegraphics
    [width=0.87\textwidth,trim=80 0 80 0,clip]{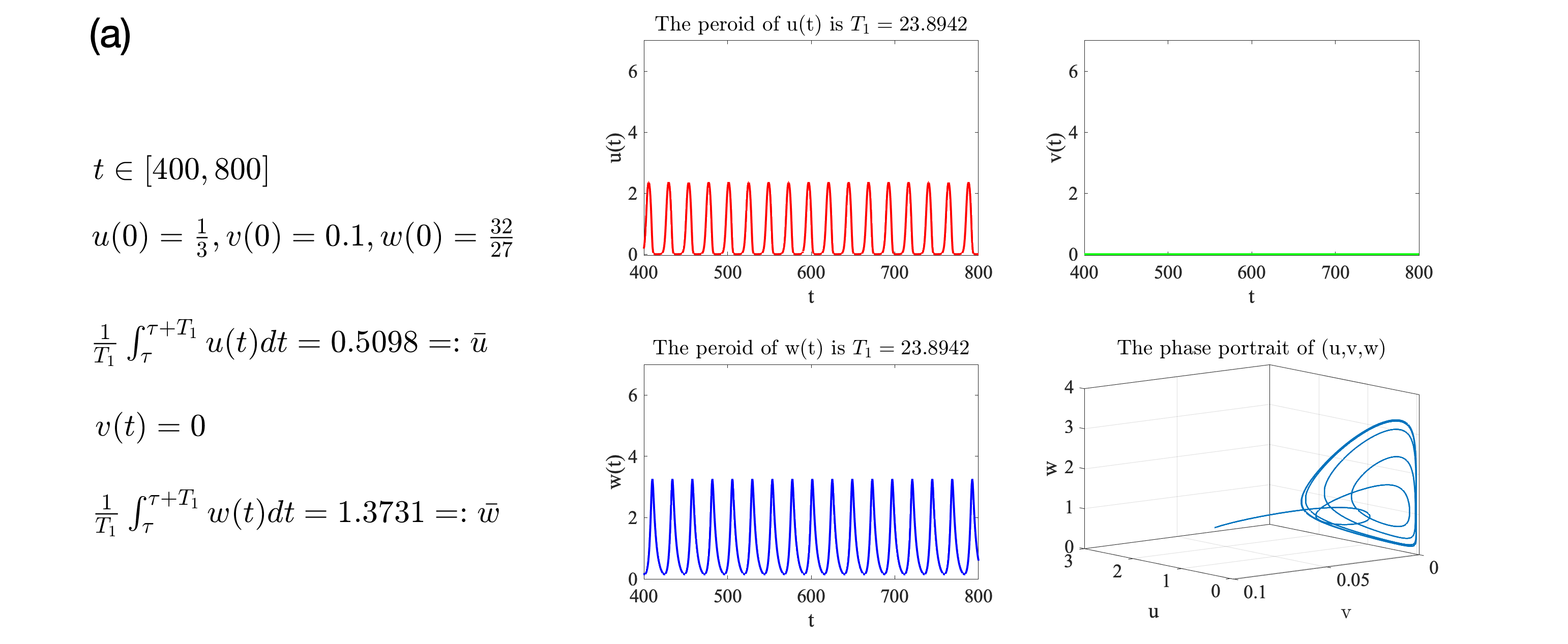}
    \vskip2mm
    \includegraphics
    [width=0.87\textwidth,trim=80 0 80 0,clip]{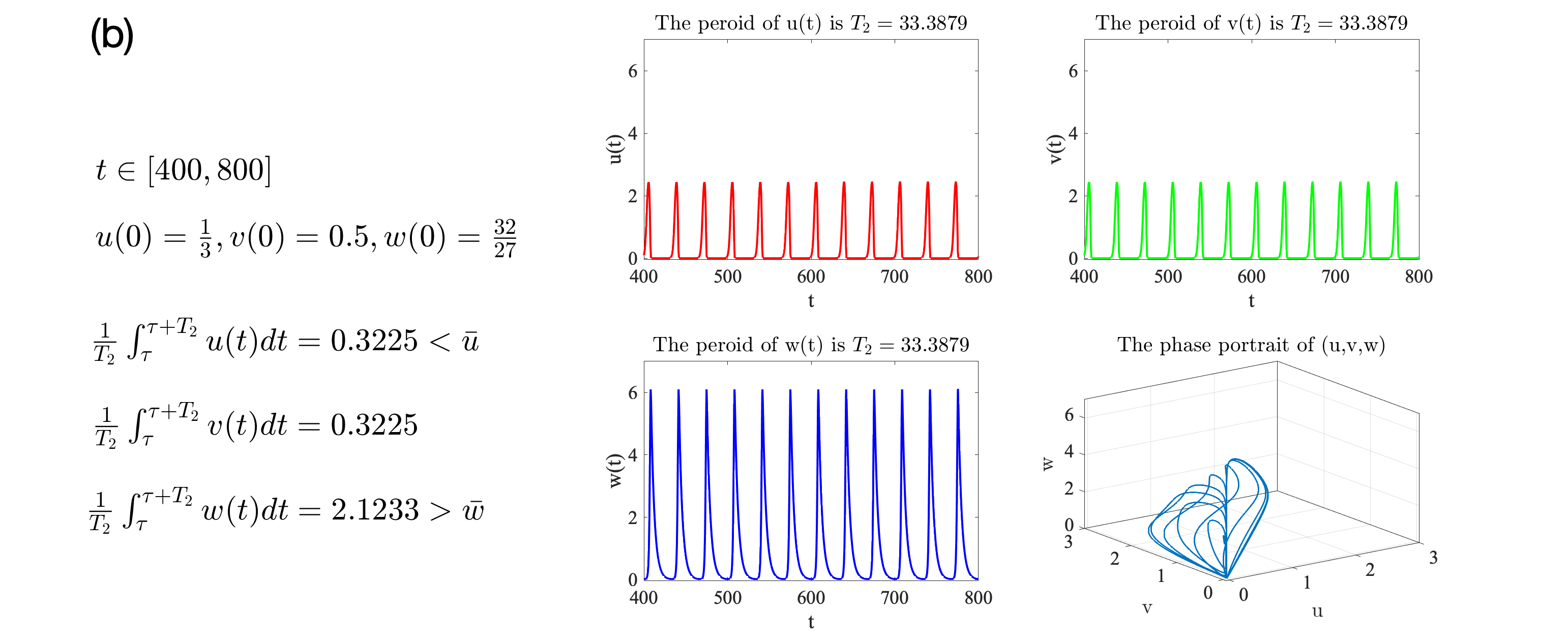}
    \vskip2mm
    \includegraphics
    [width=0.87\textwidth,trim=80 0 80 0,clip]{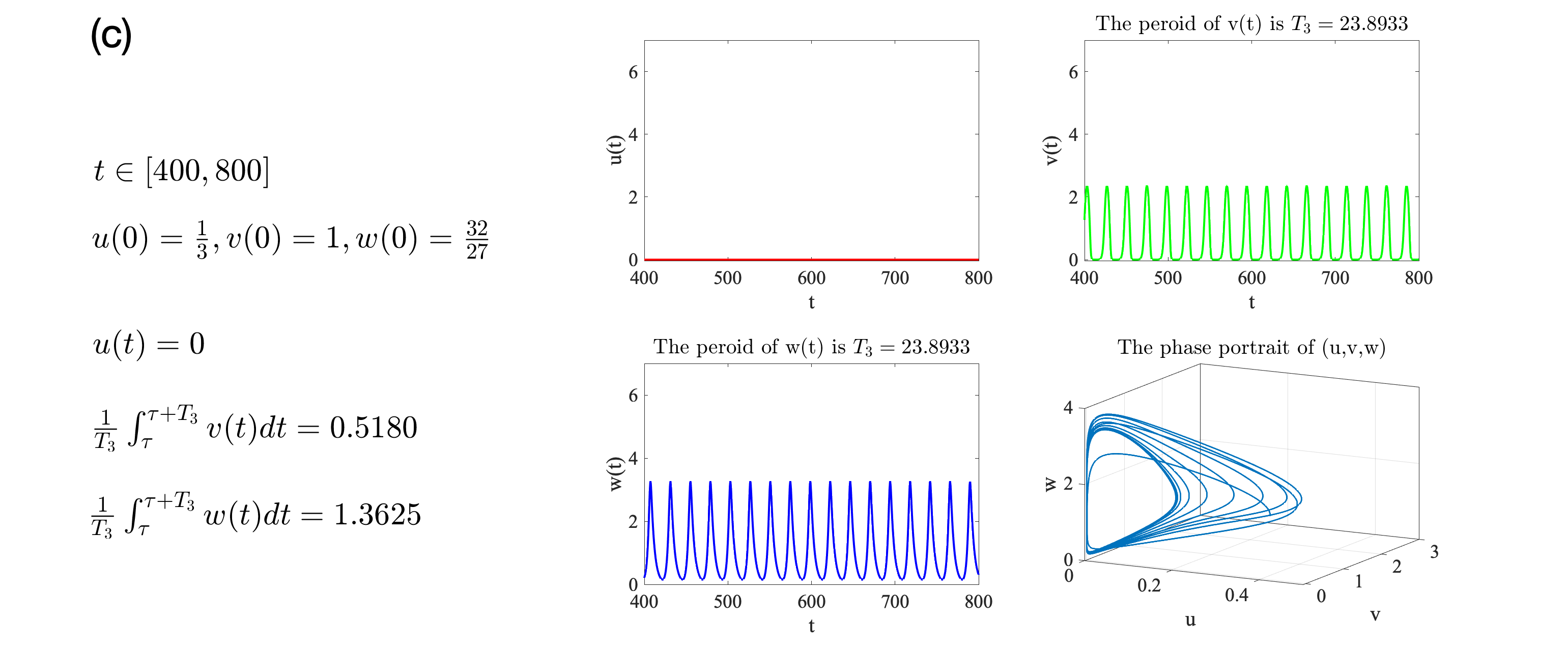}
    \caption{\small Asymptotic dynamics of {the} system \eqref{model} with \eqref{eqh2} under the parameter setting \eqref{eq4.3} and $\theta=\frac14$. The initial data are taken as :
    \text{(a)} $(\frac{1}{3},0.1,\frac{32}{27})$;
    \text{(b)} $(\frac{1}{3},0.5,\frac{32}{27})$;
    \text{(c)} $(\frac{1}{3},1,\frac{32}{27})$.
    }
    \label{fig3}
\end{figure}

With the stability results given in {Lemmas \ref{lem4.1}-\ref{lem4.4}}, we summarize the stability/instability properties of all equilibria in Table \ref{table4}. The bifurcation diagrams of these equilibria are shown in Fig. \ref{fig2}. The results in Table \ref{table4} imply that if the predator's mortality rate $\theta$ is sufficiently large ($\theta\geq \frac{3}{2}$), then the predator will die out and two prey species coexist (i.e. $E_{uv}$ is globally asymptotically stable). If $\theta$ is suitably large (i.e. $\theta \in [\frac{4}{3},\frac{3}{2})$), then the predator coexist with the two prey species (i.e. ${Q_*^0}$ is globally asymptotically stable). However, if $\theta$ is not large (i.e. $0<\theta<\frac{4}{3}$),  the global dynamics largely remain unknown and different outcomes are expected from the local dynamics shown in Table 5. We shall use numerical simulations to foresee the possible global dynamics for $0<\theta<\frac{4}{3}$ and quantify the population size in the next subsection and discuss the underlying biological implications. Our numerical simulations and biological discussion will focus on the questions {\bf A1} and {\bf A2} given in the Introduction. Therefore we consider two classes of initial data. The first class of initial data {is} set as a perturbation of the invasive species free equilibrium $Q_1=(u_{Q_1}, 0, w_{Q_1})$ while keeping $u_{Q_1}$ and $w_{Q_1}$ unchanged, namely $(u_0, v_0,w_0)=(u_{Q_1}, R, w_{Q_1})$ with $R>0$ being a constant. The numerical results for such initial data can address the effect of the invasion of the invasive prey species on the dynamics of the native prey species, and further investigate under what conditions the native prey species is reduced in its population size or annihilated.  The second class of initial {data is} set as a perturbation of coexistence equilibrium $Q_*$, for which the numerical results can address the robustness of the coexistence in the predator-mediated apparent competition.\\

\noindent {\bf Numerical simulations and implications}.
The numerical simulations for $\theta\in(0,\frac43)$  will be divided into three {parts}: $\theta\in (0,\frac12)$, $\theta\in[\frac12,\frac34)$ and $\theta\in[\frac34,\frac43)$, and in each part we take an arbitrary value of $\theta$ to conduct the numerical simulations.

\textbf{Part 1:  $\theta \in (0,\frac12)$}.
We take $\theta=\frac14\in (0,\frac12)$ and focus on the {semi-coexistence} equilibria $Q_1=(\frac13,0,\frac{32}{27})$ given by \eqref{eq4.5} which is unstable (see Table \ref{table4}). The initial value is set as  $(u_0, v_0, w_0)=(\frac13,R,\frac{32}{27})$ with $R>0$ denoting the initial mass of invasive prey species $v$. The numerical results for different values of $R$  are plotted in Fig. \ref{fig3}, where we find three different typical outcomes showing whether the invasion is successful depends on the initial biomass of invasive prey species $v$ if the mortality rate of the predator is suitably small. Specifically, we have the following observations.
\begin{itemize}
    \item[(i)] If the initial mass $v_0$ of the invasive prey species is small (e.g. $v_0=R=0.1$), then the invasive prey species fails to invade and die out while the native prey species coexist with the predator periodically (i.e. the solution asymptotically develops into a periodic solution $(u_1^*(t), 0, w_1^*(t)$   with period $T_1=23.8942$); see Fig. \ref{fig3}\text{(a)}.
    \item[(ii)] If the initial mass $v_0$ of the invasive prey species is medial (e.g. $v_0=R=0.5$),  the invasive species $v$ invades successful and finally coexists with the native prey species $u$ and the predator $w$ periodically (i.e. the solution asymptotically develops into a periodic solution $(u_2^*(t), v_2^*(t), w_2^*(t)$ with period $T_2=33.3879$), but the biomass of the native prey species $u$ is reduced  due to the increase of the predator's biomass, where
    \begin{align*}
    \begin{cases}
        \medskip
        \frac{1}{T_1} \int_0^{T_1} u_1^*(t) d t=\bar u=0.5098>0.3225=\frac{1}{T_2} \int_0^{T_2} u_2^*(t) d t,\\
        \frac{1}{T_1} \int_0^{T_1} w_1^*(t) d t=\bar w=1.3625<2.1233=\frac{1}{T_2} \int_0^{T_2} w_2^*(t) d t,
    \end{cases}
\end{align*}
as shown in Fig. \ref{fig3}\text{(b)}.
    \item[(iii)] If the initial mass $v_0$ of the invasive prey species is large (e.g. $v_0=R=1$),  the invasive species $v$ not only invade successfully but also wipes out the native prey species via the predator-mediated apparent competition (i.e. the solution asymptotically develops into a periodic solution $(0, v_3^*(t), w_3^*(t)$ with period $T_2=23.8933$)); see Fig. \ref{fig3}\text{(c)}.
\end{itemize}

The above observations indicate that whether the invasive prey species can invade successfully to trigger the predator-mediated apparent competition essentially depends on the size of the initial biomass of the invasive prey species. Small initial biomass will lead to failed invasions and does not change the existing population dynamics. However, if the invasive prey species has a suitably large initial biomass, then the invasion will be successful and the predator-mediated apparent competition will take effect, resulting in the decrease or even extinction of the native prey species. To reduce the biomass of a certain species (like pests), it is suitable to employ the strategy of predator-mediated apparent competition by introducing a new (invasive) species with appropriate initial biomass.

\begin{figure}[!hb] \centering
    \includegraphics
    [width=0.8\textwidth,trim=90 100 75 0,clip]{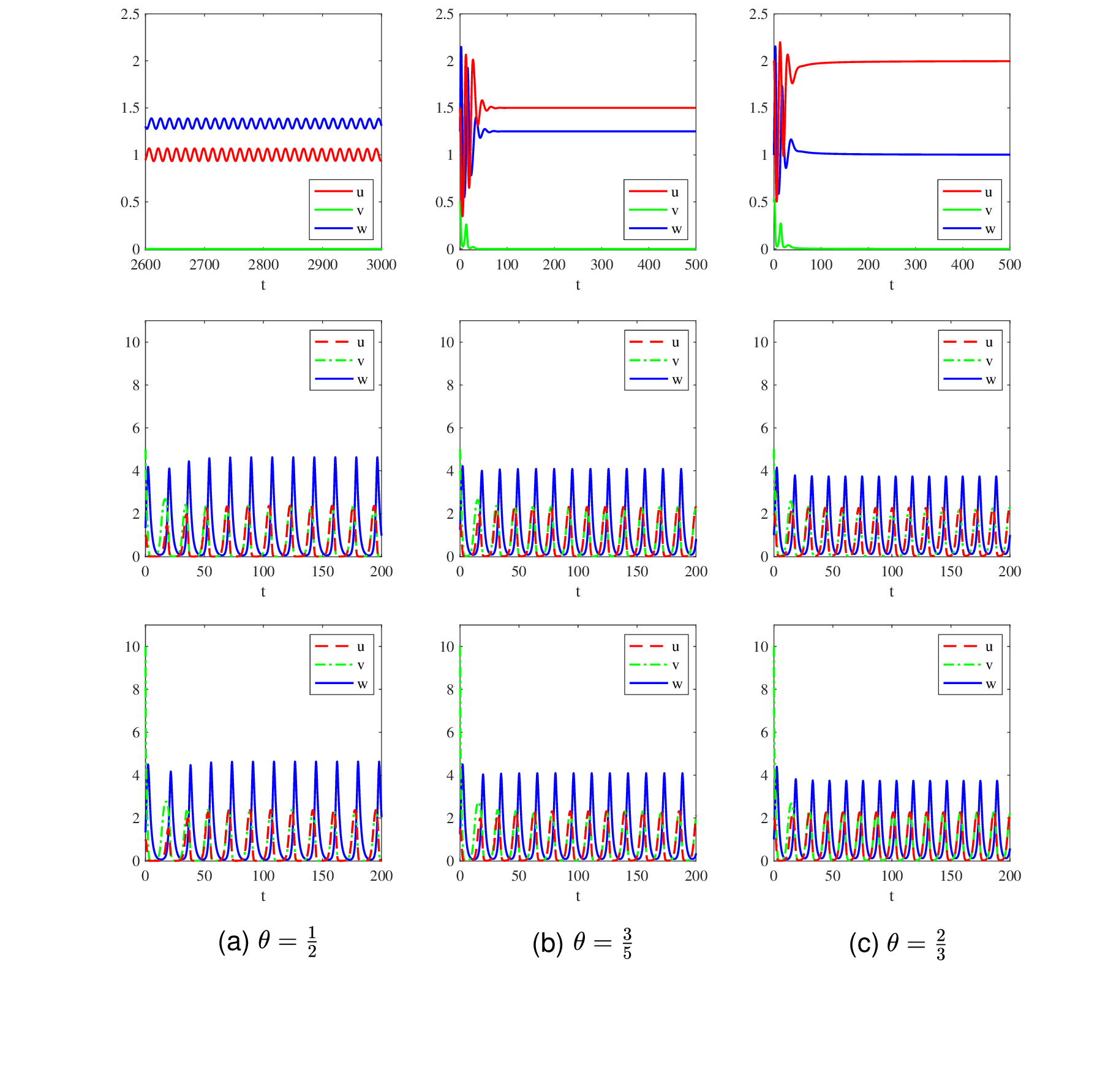}
    \caption{\small
    Long-time dynamics of {the} system \eqref{model} with \eqref{eqh2}{, \eqref{eq4.3}, and different values of $\theta\in\dkh{\frac12,\frac35,\frac23}$}. The initial data are taken as  $(u_0,v_0,w_0)=Q_1+(0,R,0)$, where $Q_1=(1,0,\frac43)$ in (a), $Q_1=(\frac32,0,\frac54)$ in (b), and  $Q_1=(2,0,1)$ in (c); $R=0.5$ in the first row, $R=5$ in the second row, and $R=10$ in the third row.
    }
    \label{fig4}
\end{figure}

\textbf{Part 2:  $\theta\in[\frac12,\frac34)$}.
In this case, we  first consider three values for
\begin{align}\nonumber
    \theta\in\dkh{\frac12,\frac35,\frac23}
\end{align}
and corresponding numerical simulations are plotted in Fig. \ref{fig4}. We observe similar behaviors to those for $\theta \in (0,\frac{1}{2})$ shown in Fig. \ref{fig3}, where the invasive species $v$ will fail to invade if its initial mass is small as illustrated in the first row of Fig. \ref{fig4}. However, with a large initial mass, the invasive prey species can invade successfully as shown in the second row of Fig. \ref{fig4}, but cannot annihilate the native prey species via the predator-mediated apparent competition. This is perhaps because the predator's mortality rate $\theta$ is too large to annihilate the native species even if the invasive species can boost the food supply of the predator. This alongside the numerical simulations shown in Fig. \ref{fig3} implies that whether the native prey species will be driven to extinction via the predator-mediated apparent competition depends not only on the initial mass of the invasive species but also on the mortality rate of the predator. Further increasing the value of $\theta$ to be $\theta=\frac7{10} \in (\theta_1, \frac{3}{4})$, at which $Q_1=(\frac73,0,\frac{20}{27})$, we find from the numerical simulations shown in Fig. \ref{fig5}(a) that the invasion is successful albeit small initial population abundance of the invasive species (in comparison with those in the first row of Fig. \ref{fig4}). Mathematically this is because $Q_1$ is a saddle and any small perturbation of $Q_1$ will result in instability.
With a large predator's mortality rate, the invasive species (even with a large initial mass) cannot drive the native species to extinction (see Fig. \ref{fig5}\text{(b)}), similar to other large values of $\theta$ shown in the second and third rows of Fig. \ref{fig4}.
This implies if the predator has a large mortality rate, it can not drive the native prey species to extinction even if its food supply is boosted by the invasive prey species.

\begin{figure}[!h] \centering
    \includegraphics
    [width=0.65\textwidth,trim=60 0 50 10,clip]{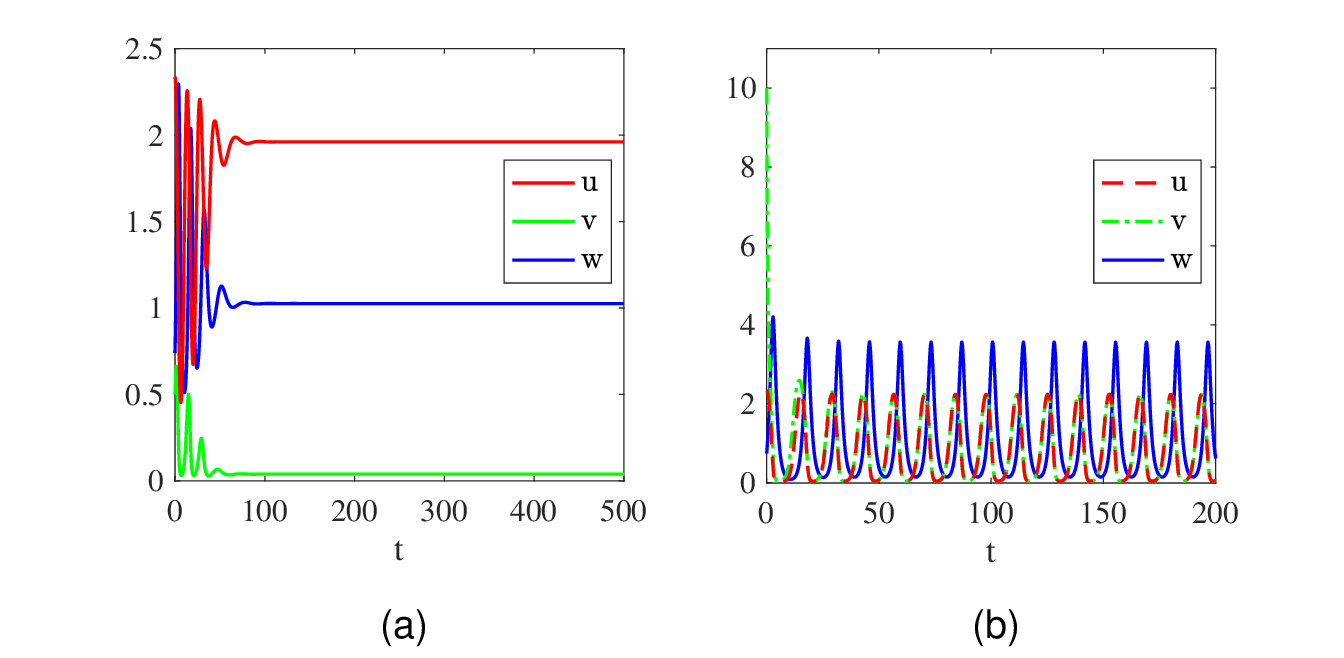}
    \caption{\small
    Long-time dynamics of {the} system \eqref{model} with \eqref{eqh2} and parameters given in \eqref{eq4.3} for $\theta=\frac7{10}$. The initial data are taken as
    $(u_0,v_0,w_0)=Q_1+(0,R,0)$, where $Q_1=(\frac73,0,\frac{20}{27})$, $R=0.5$ for \text{(a)} and $R=10$ for \text{(b)}.
    }
    \label{fig5}
\end{figure}

Concerning the questions raised in {\bf A1}, the above numerical results pinpointed two key factors determining the successful invasion of the invasive prey species: initial invasive mass $v_0$ and mortality rate $\theta$ of the predator. Specifically for a fixed mortality rate $\theta$ not large, increasing the initial invasive mass $v_0$ can lead to a successful invasion. If the mortality rate $\theta$ is large, then the predator will go extinct and the mass of the native prey species will not be affected though the invasion is successful. Conversely, for a fixed initial invasive mass that is not too small, the larger mortality rate of the predator will be beneficial to the success of the invasion. Moreover, the population abundance of the native prey species will be reduced by the predator-mediate apparent competition as shown in Fig. \ref{fig3}. Another interesting finding in our numerical simulations is that the asymptotic profiles of the native and invasive prey species coincide as long as the non-trivial periodic coexistence appears (see Fig. \ref{fig3} to Fig. \ref{fig5}). This is a mysterious mathematical question deserving further investigation.

Next, we explore how the population abundance of native prey species changes with respect to the initial invasive mass. To this end, we take the numerical results shown in  Fig. \ref{fig4}\text{(b)} as an example.    {Denote the three solutions shown in Fig. \ref{fig4}\text{(b)} by $(u_{R}^*,v_{R}^*,w_{R}^*)(t)$ for $\theta=\frac35$ and $R=0.5,5,10$. Then $(u_{R}^*,v_{R}^*,w_{R}^*)(t)\mid_{R=0.5}\equiv Q_1=(\frac32,0,\frac54)$ for all $t>0$, and $(u_{R}^*,v_{R}^*,w_{R}^*)(t)$ are periodic solutions with the period $T_R$ for $R=5,10$. Quantitative estimates of the total population in a period for $R=0.5,5,10$ are summarized in Table \ref{table5}.} We see from the results shown in Table \ref{table5} that the total mass of the native prey species decreases with respect to the initial mass of the invasive prey species, as expected.

{
\renewcommand{\arraystretch}{1.5}
\begin{table}[!h]
    \caption{{Quantitative properties of $(u_{R}^*,v_{R}^*,w_{R}^*)(t)$ for $R=0.5,5,10$.}}
    \begin{tabular}
        {m{4cm}<{\centering}|m{2.5cm}<{\centering}|m{2.5cm}<{\centering}|m{2.5cm}<{\centering}}
        \hline
         $R$
        &  0.5
        &  5
        &  10
        \\
        \hline
         Period $T_R$
        & /
        & 15.3714
        &15.3714
        \\
        \hline

        $\bar{u}=\frac{1}{T_R} \int_0^{T_R} u_{R}^*(t) d t$
        &  $\frac32$
        &  0.6277
        &  0.6275
        \\
        \hline

        $\bar{v}=\frac{1}{T_R} \int_0^{T_R} v_{R}^*(t) d t$
        &  0
        &  0.6277
        &  0.6275
        \\
        \hline

        $\bar{w}=\frac{1}{T_R} \int_0^{T_R} w_{R}^*(t) d t$
        & $\frac54$
        & 1.5866
        & 1.5844
        \\
        \hline
    \end{tabular}
    \label{table5}
    \begin{adjustwidth}{19mm}{8mm}
        \begin{tablenotes}
            \footnotesize
            \item {Remark: Here the notation ``/" means “this is not a  non-constant periodic case”.}
        \end{tablenotes}
    \end{adjustwidth}
\end{table}
}

\begin{figure}[!h] \centering
    \includegraphics
    [width=0.9\textwidth,trim=90 70 70 40,clip]{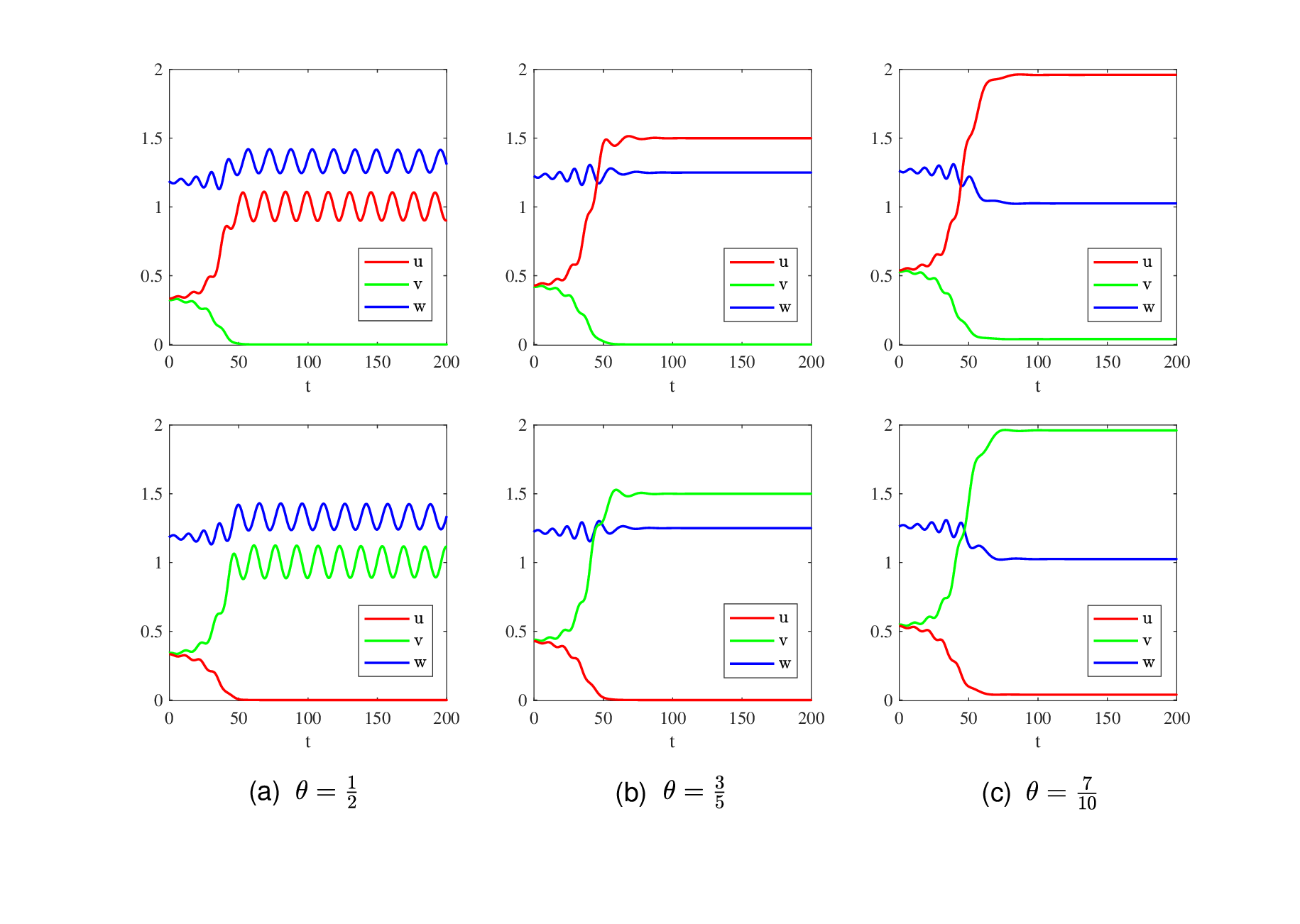}
    \caption{\small
        Long-time dynamics of {the} system \eqref{model} with \eqref{eqh2} and parameters given in \eqref{eq4.3} for $\theta\in\dkh{\frac12,\frac35,\frac7{10}}$. The initial data are taken as
        $(u_0,v_0,w_0)={Q_*^0}+(0,R,0)$, where $R=-0.01$ for the first row and $R=0.01$ for the second row, and ${Q_*^0}$ is given by \eqref{eq4.6}:
    \text{(a)} $(\frac{1}{3},\frac{1}{3},\frac{32}{27})$;
    \text{(b)} $(\frac{3}{7},\frac{3}{7},\frac{60}{49})$;
    \text{(c)} $(\frac{7}{13},\frac{7}{13},\frac{640}{507})$.
    }
    \label{fig6}
\end{figure}
\begin{figure}[!ht] \centering
    \includegraphics
    [width=0.9\textwidth,trim=100 10 75 10,clip]{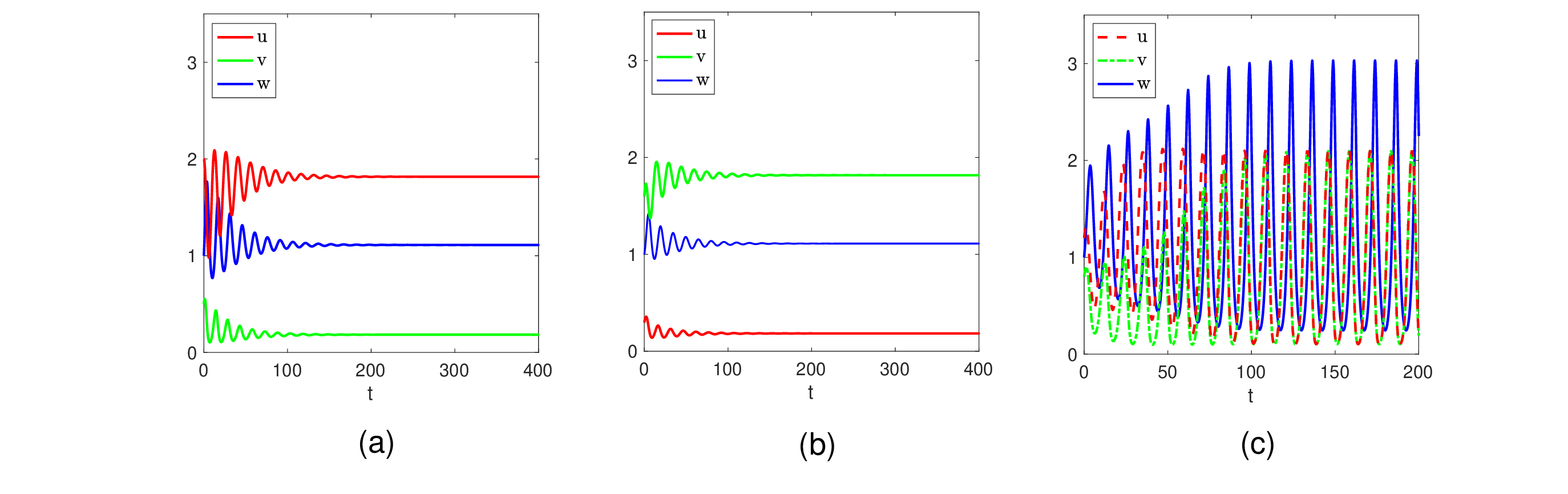}
    \caption{\small
        Long-time dynamics of {the} system \eqref{model} with \eqref{eqh2}  under the parameter setting \eqref{eq4.3} and $\theta=0.8$. The initial data are taken as $(u_0,v_0,w_0)$:
    \text{(a)} $(2,0.5,1)$;
    \text{(b)} $(0.3,1.6,1)$;
    \text{(c)} $(1.2,0.8,1)$.
    }
    \label{fig7}
\end{figure}

We proceed to examine whether the constant coexistence/positive solution is stable. To this end, we shall investigate the stability/instability of $Q_0^*$ which exists if $\theta<\frac{3}{2}$. The results of Theorem 2.2 have asserted that $Q_0^*$ is globally asymptotically stable if $\theta \in [\frac{4}{3},\frac{3}{2})$. This indicates that if the mortality rate of the predator is appropriately large, then the coexistence will persist as long as the invasion is successful. However, this is no longer the case if the mortality rate of the predator is suitably small, as shown in Fig. \ref{fig6} where we see that any small negative (resp. positive) perturbation of one prey species density may lead to the extinction or abundance decrease of this species (reps. the other one). This indicates that the constant coexistence solution is not robust against (small) perturbations.

{\textbf{Part 3: $\theta\in[\frac34,\frac43)$}.}
In view of Table \ref{table4}, both {coexistence equilibria} ${Q_*^1}$ and ${Q_*^2}$ are stable for $\theta\in[\frac34,1)$, that is {the} system \eqref{model} generates a bistable dynamics as illustrated in Fig. \ref{fig7}  for $\theta=0.8$, where
$${Q_*^1}=\xkh{1+\sqrt{\frac23},1-\sqrt{\frac23},\frac{10}9}, \ {Q_*^2}=\xkh{1-\sqrt{\frac23},1+\sqrt{\frac23},\frac{10}9}.$$ With an initial value $(u_0,v_0,w_0)$ which is``closer" to ${Q_*^1}$ than ${Q_*^2}$, the corresponding numerical results shown Fig. \ref{fig7}\text{(a)} demonstrate that the solution converges to ${Q_*^1}$, while Fig. \ref{fig7}\text{(b)} illustrates the convergence of solutions to ${Q_*^2}$ when the initial value is closer to ${Q_*^2}$. We wonder if a non-constant solution may develop if the initial value is not close to either of these two stable equilibria.
Hence we choose an initial value $(u_0,v_0,w_0)=(1.2,0.5,1)$ neither close to ${Q_*^1}$ nor to ${Q_*^2}$, the corresponding numerical result shown in Fig. \ref{fig7}\text{(c)} demonstrates that the periodic solution will develop. But how to rigorously prove the existence of periodic solutions remains an interesting open question.

In applications, the invasive prey species may be used as a biological control agent to regulate the population size of the native prey species if they are harmful (like pests). The ideal situation is that a small number of invasive prey species can achieve this goal. The above linear stability analysis alongside numerical simulations indicates that this is unfeasible if two prey species are ecologically identical (i.e. the symmetric case). However, this is achievable when two prey species are ecologically different (i.e. asymmetric case) as to be shown in the next subsection.
\subsection{\textbf{{Asymmetric apparent competition}}}
For simplicity, we first rescale the system \eqref{model} with \eqref{eqh2}. To this end, we set
\begin{align}\label{eq4.7}
    \widetilde u=\frac u{K_1},\
    \widetilde v=\frac v{K_2},\
    \widetilde w=w,\quad (\widetilde \gamma_i,\widetilde h_i,\widetilde \beta_i)
    =(\gamma_i,h_i K_i,\beta_i K_i),\quad  i=1,2.
\end{align}
Substituting the above rescalings into \eqref{model} with \eqref{eqh2} and dropping the tildes for brevity, we obtain the following rescaled system
\begin{equation}\label{eq4.8}
    \begin{cases}
        \medskip u_{t}= u\left(1-u \right)-w \frac {\gamma_1 u}{1+\gamma_1   h_1 u},  \quad &t>0,\\
        \medskip v_{t}= v\left(1-v  \right)-w \frac {\gamma_2 v}{1+\gamma_2   h_2 v},  \quad &t>0,\\
         \medskip w_{t}=w\xkh{ \beta_{1}  \frac {\gamma_1 u}{1+\gamma_1   h_1 u}
         + \beta_{2}\frac {\gamma_2 v}{1+\gamma_2   h_2 v}
         -\theta},  \quad &t>0,\\
         (u,v,w)(0)=(u_0,v_0,w_0).
    \end{cases}
\end{equation}
The rescaled system \eqref{eq4.8}, which can be viewed as a special case of \eqref{eq4.7} with $K_1=K_2=1$, has three types of parameters: capture rates $\gamma_i$, handling times $h_i$ and conversion rates $\beta_i$, where $i=1,2$, In the following, we shall focus on the case that the two prey species have different capture rates (i.e. $\gamma_1\ne \gamma_2$)  by assuming $h_1=h_2, \beta_1=\beta_2$ and study the effects of predator-mediated apparent competition with different capture rates. For definiteness, we set without loss of generality
\begin{align}\label{eg2}
        h_i=1,\ \beta_i=b>0\
        \quad\text{and}\quad
        0< \gamma_2< \gamma_1=1.
    \end{align}
The biological meaning of parameter values set in \eqref{eg2} is that the two prey species $u$ and $v$ have the same handling times and conversion rates but vary in capture rates, while the predator prefers to hunt the native prey species $u$ ($\gamma_1>\gamma_2$). Clearly the rescaled system \eqref{eq4.8} with {\eqref{eg2}} has four predator-free equilibria
\begin{align}\nonumber
    E_0=(0,0,0),\ E_u=(1,0,0),\ E_v=(0,1,0),\ E_{uv}=(1,1,0),\quad \text{if }\theta>0,
\end{align}
two semi-coexistence equilibria
\begin{align} \nonumber
    \begin{cases}
        \medskip
        Q_1=\xkh{u_{Q_1},0,w_{Q_1}}=\xkh{\frac{\theta }{b-\theta},0,\frac{b (b-2 \theta )}{(b-\theta )^2}},
        \quad \ &\text{if } \theta\in\xkh{0, L_1},\\
        Q_2=\xkh{0,v_{Q_2},w_{Q_2}}=\xkh{0,\frac{\theta }{\gamma_2(b-\theta)},\frac{b (b \gamma_2-(1+\gamma_2) \theta )}{\gamma_2^2(b-\theta )^2}},
        \quad \ &\text{if } \theta\in\xkh{0, L_2},
    \end{cases}
\end{align}
and a unique coexistence equilibrium (see Lemma \ref{lemB.3} in Appendix B for detailed reasons)
\begin{align}\label{eq4.9}
    Q_*=(u_*,v_*,w_*),\quad \text{if } \theta\in(\Theta_1, L),
\end{align}
where
\begin{align} \label{eq4.11}
    \begin{cases}
        \medskip
        L_1=\frac b2> L_2=\frac{b \gamma_2}{1+\gamma_2},\  L= L_1+ L_2<b,\\
        \Theta_1=\varphi_1(\gamma_2) b\in(0, L_2),\quad
        \varphi_1(\gamma_2):=\frac{ \sqrt{(1-\gamma_2) (3 \gamma_2+1)}-(1-\gamma_2) (2 \gamma_2+1)}{2 \gamma_2^2}.
    \end{cases}
\end{align}
For $b>0$ and $\gamma_2\in(0,1)$, it holds that
\begin{align*}
    \begin{cases}
        \medskip
        \varphi_1''(\gamma_2)<0,\ \varphi_1'(\frac23)=0,\
    \lim\limits_{\gamma_2\rightarrow0}\varphi_1(\gamma_2)=\lim\limits_{\gamma_2\rightarrow1}\varphi_1(\gamma_2)=0,\\
    0<\Theta_1\leq b\varphi_1(\frac23)=\frac b4, \text{ and }\Theta_1\text{ attains its maximum }\frac b4 \text{ if and only if }\gamma_2=\frac23.
    \end{cases}
\end{align*}
This implies that $\Theta_1$ is non-monotone in $\gamma_2$, but a convex function maximized at $\gamma_2=\frac23$.

\begin{rem}\label{rem4.3}
    Applying Theorem \ref{thm2.2} (iii)-(iv) with $K_1=K_2=1$ to system \eqref{eq4.8}-\eqref{eg2}, we can easily find that $Q_*$ is globally asymptotically stable for $\theta\in(\Theta_1, L)$, and $E_{uv}=(1,1,0)$ is globally asymptotically stable for $\theta\geq  L$. Since $\frac {1+\gamma_1 h_1}{\gamma_1}=2>\lim\limits_{\theta\rightarrow \Theta_1}w_{Q_2}=1$ for $b>0$ and $\gamma_2\in(0,1)$, the results in Theorem \ref{thm2.2}(ii) with $K_1=K_2=1$ are inapplicable to assert the global stability of $Q_2$ for $\theta\in(0,\Theta_1]$. However,  this can be shown in the following lemma.
\end{rem}
\begin{lem}\label{lem4.5}
    The semi-coexistence equilibrium $Q_2$ of the rescaled system \eqref{eq4.8} with \eqref{eg2} is globally asymptotically stable if $\theta\in(0,\Theta_1]$.
\end{lem}
\begin{proof}
    Let $\theta\in(0,\Theta_1]$. Then \eqref{eq4.11} implies $0<\theta<\frac{b \gamma_2}{1+\gamma_2}<\frac b2$. For $t>0$, let
    \begin{equation}\nonumber
        \medskip
        \mathcal E(t;Q_2)=bu
        +(b-\theta) \xkh{v-v_{Q_2}-v_{Q_2}\ln\frac v{v_{Q_2}}}
        +\xkh{w-w_{Q_2}-w_{Q_2}\ln\frac w{w_{Q_2}}}.
    \end{equation}
    Then by similar arguments as in the proofs of Lemma \ref{lem3.2} and Lemma \ref{lem3.5}, we have $\mathcal E (t;Q_2)>0\text{ for all }(u,v,w)\neq Q_2$, and
    \begin{align*}
        \mathcal E'(t;Q_2)
        \deyu ~ b \left(1-u -\frac {w }{1+ u}\right)u+(b-\theta)\left(1-v  - \frac {\gamma_2 w}{1+\gamma_2  v}\right)(v-v_{Q_2})\non
        &+\xkh{\frac{b \gamma_2 v}{\gamma_2 v+1}+\frac{b u}{u+1}-\theta}(w-w_{Q_2})\non
        \deyu -(b-\theta)\frac{(1-\gamma_2+\gamma_2 (v+v_{Q_2}))}{\gamma_2 v+1}(v-v_{Q_2})^2
        -\frac{b u^3}{u+1}+\frac{b u \varphi_2(\theta )}{\gamma_2^2 (u+1) (b-\theta )^2}\non
        <& -(b-\theta)\frac{\gamma_2 (v+v_{Q_2})}{\gamma_2 v+1}(v-v_{Q_2})^2
        -\frac{b u^3}{u+1},
    \end{align*}
    where we have used $\gamma_2\in(0,1)$ and the fact that the quadratic function
    \begin{align}\label{addeq1}
        \varphi_2(\theta ):=&~\gamma_2^2\theta ^2 +b   \left(-2 \gamma_2^2+\gamma_2+1\right)\theta+b^2 (\gamma_2-1) \gamma_2\non
        =&~\gamma_2^2\zkh{\theta+\Theta_1+b \left(\frac{1+\gamma_2}{\gamma_2^2}-2\right)}(\theta-\Theta_1)
    \end{align}
    is nonpositive for $\theta\in(0,\Theta_1]$ in the last inequality. Finally, similar arguments based on the Lyapunov function method and LaSalle's invariant principle as in the proof of Lemma \ref{lem3.2} complete the proof.
\end{proof}

With Remark \ref{rem4.3} and Lemma \ref{lem4.5}, we summarize the global stability results in Table \ref{table6} for the rescaled system \eqref{eq4.8} with \eqref{eg2}.

{\renewcommand{\arraystretch}{1.5}
\begin{table}[H]
    {
        \caption{{Global stability of \eqref{eq4.8} with \eqref{eg2}.}}
    \begin{tabular}
        {m{3.5cm}<{\centering}|
        m{3cm}<{\centering}|
        m{3cm}<{\centering}|
        m{3cm}<{\centering}}
        \hline
        $\theta$
        & $(0,\Theta_1]$
        & $(\Theta_1,L)$
        & $[L,\infty)$\\
        \hline
        Global stability
        & $Q_2\text{ is GAS}$
        & $Q_*\text{ is GAS}$
        & $E_{uv}\text{ is GAS}$\\
        \hline
    \end{tabular}
    \label{table6}
    \begin{adjustwidth}{13mm}{19mm}
        \begin{tablenotes}
            \footnotesize
            \item Note: The notation ``GAS" has the same interpretation as in Table \ref{table2}. The parameter $\Theta_1$ is given in \eqref{eq4.11}.
        \end{tablenotes}
    \end{adjustwidth}
    }
\end{table}
}
Under the parameter setting {\eqref{eg2}}, the capture rate of the invasive prey species $v$ is smaller than the native prey species $u$, namely {$0<\gamma_2<\gamma_1=1$}. According to the results shown in Table \ref{table6} for any $\theta>0$, we can derive the following biological implications.
\begin{itemize}
    \item[(i)] {If} $\theta \in(0, \Theta_1]$ (i.e. the predator has a low mortality rate), the global stability of $Q_2$  implies that the invasive prey species can invade successfully regardless of its initial population size and wipe out the native prey species via the predator-mediated apparent competition.
    \item[(ii)] {If} $\theta \in(\Theta_1, L)$ (i.e., the predator has a moderate mortality rate), then the global stability of $Q_*$  indicates that moderate predator's mortality allows the native prey species to survive and to coexist with the invasive prey species and the predator.
    \item[(iii)] {If} $\theta \geq { L}$, the global stability of $E_{uv}$ entails that the poor physical condition of the predator (i.e. the predator has a large mortality rate) will result in the extinction of the predator though the invasive prey species can boost the food supply to the predator.
\end{itemize}
The above results indicate that if the predator has a hunting preference for the native species (i.e. larger capture rate of the native prey species), then the invasive prey species can always invade successfully regardless of its initial population size. Furthermore, whether the native prey species can be eradicated through the predator-mediated apparent competition essentially depends upon the mortality rate of the predator (i.e. low predator mortality rate will result in the extinction of the native prey species while a moderate or large mortality rate will allow the native prey species to persist). In the general parameter set in which $0<\gamma_2<\gamma_1$,  the case $0<\gamma_2<\gamma_1=1$ is only a special situation where we can completely classify the global stability of solutions as given in Table \ref{table6}. For other parameter regimes contained in the set $0<\gamma_2<\gamma_1$, we can perform the linear stability analysis to obtain local stability results and employ the Lyapunov function method alongside LaSalle's invariant principle to obtain the global stability results in partial parameter regimes, but a complete classification of global stability can not be established. Indeed, in some parameter regimes, the periodic solutions may exist (see Fig. \ref{fig8}), and hence the global stability in the whole parameter regime is impossible.  Nevertheless, the biological phenomena observed from our numerical simulations (not shown here for brevity) are essentially similar to the case $0<\gamma_2<\gamma_1=1$: the invasive prey species will always invade successfully regardless of its initial population abundance and can even wipe out the native prey species through the predator-mediated apparent competition if the mortality rate of the predator is low,  while the native prey species can persist and coexist with the predator and invasive prey species if the mortality rate of the predator is moderate, where the difference from the case $0<\gamma_2<\gamma_1=1$ is that the coexistence state may be periodic besides constant as shown in Fig. \ref{fig8}.

\begin{figure}[!ht] \centering
    \includegraphics
    [width=1\textwidth,trim=160 10 120 20,clip]{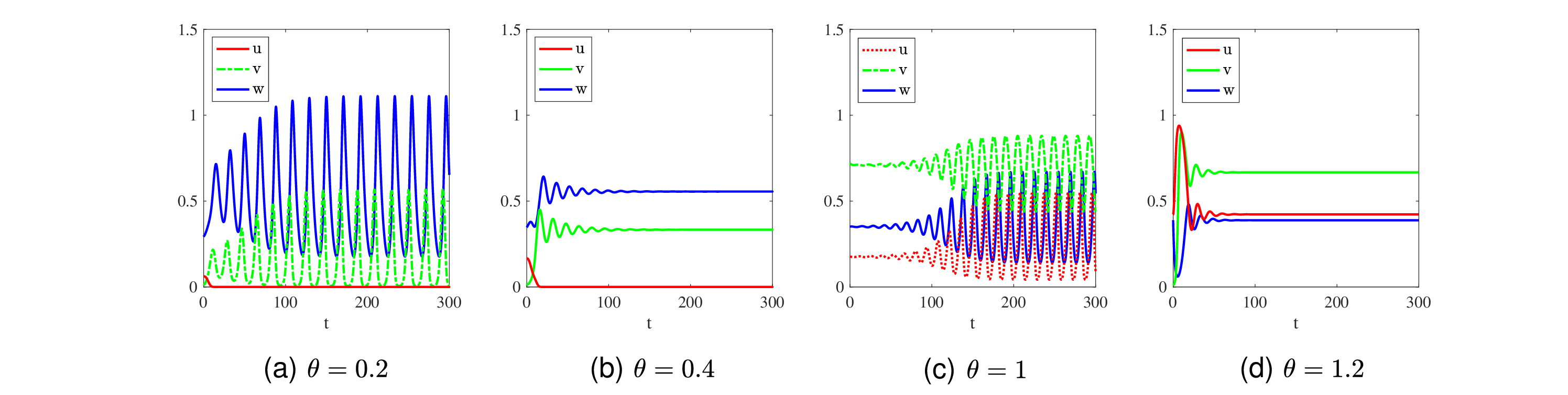}
    \caption{\small
    {Long-time dynamics of {the} rescaled system \eqref{eq4.8} with $ h_i= \beta_i=1$ ($i=1,2$), $(\gamma_1,\gamma_2)=(4,2)$ and $\theta=0.2,0.4,1,1.2$. The initial data are taken as
    $(u_0,v_0,w_0)=Q_1+(0,0.01,0)$ in \text{(a)}-\text{(b)}, $(u_0,v_0,w_0)=Q_*+(0,0.01,0)$ in \text{(c)}, and $(u_0,v_0,w_0)=(u_*,0.01,w_*)$ in \text{(d)}.
    }
    }
    \label{fig8}
\end{figure}

If we assume $0<\gamma_1<\gamma_2=1$ (i.e. the capture rate of the native prey species is {smaller} than the invasive prey species), then the results in Table \ref{table6} hold by swapping $Q_1$ with $Q_2$. This means that if the predator has a hunting preference for the {invasive} prey species, then a successful invasion depends heavily on the predator's mortality rate (precisely the invasion will fail for $\theta\in(0, \Theta_1]$ while succeeding for $\theta > \Theta_1$). Even if the invasion is successful, the invasive prey species is unable to wipe out the native prey species through predator-mediated apparent competition, regardless of its initial population abundance. These interesting results have significant value in applications. For instance, if we were to control the population abundance of some harmful species (like pests) by their natural enemies, we can introduce a small amount of secondary (invasive) prey species that are less preferred by their natural enemies based on the principle of predator-mediated apparent competition.


\section{Summary and discussion}\label{sec5}

The predator-mediated apparent competition is an indirect and negative interaction between two prey species mediated by a shared predator. As stressed in \cite{SKB2018E0}, quantifying such indirect effects is methodologically challenging but important for understanding ecosystem functioning. To study the effects of the predator-mediated apparent competition on population dynamics, in this paper, we propose to consider system \eqref{model} by viewing $u$ as a native prey species and $v$ as an invasive prey species, both of which share one predator $w$. We find conditions for the local and global stability of the equilibria of system \eqref{model} with Holling type I and II functional response functions,  and employ the numerical simulations to demonstrate the possible population dynamics and biological consequences due to the predator-mediated apparent competition. Among others,  we have the following observations.
\begin{itemize}
\item In the case of Holling I functional response function, the global dynamics of \eqref{model} can be completely classified (see Theorem \ref{thm2.1}), and any solution will globally asymptotically stabilize into an equilibrium as summarized in Table \ref{table2}. The results imply that the predator will die out if its mortality rate is too large (i.e., $\theta\geq L$) while coexists with the two prey species if its mortality rate is moderate (i.e. $\theta_0<\theta<L$). However, if the predator's mortality rate is small (i.e. $\theta\in(0,\theta_0]$), then the global dynamics depend on the capture rate of the prey species: the prey species with a smaller capture rate will annihilate other prey species with larger capture rates via the predator-mediated apparent competition.  In other words, the invasive species will not be able to invade successfully if the predator preferably captures them. On the contrary,  the invasive prey species will successfully invade and drive the native species to go extinct if the predator preferably hunts the native prey species. If the two prey species have the same capture rate ($\alpha_1=\alpha_2$), then they either coexist with or eliminate the predator depending on the strength of the predator's mortality rate (see Table \ref{table2}). In this scenario, the invasive prey species cannot drive the native prey species to extinction via predator-mediated apparent competition.  These results reveal that the predator's mortality rate and prey's capture rates are key factors determining the global dynamics with the Holling I functional response.

\item In the case of Holling type II functional response function, the global dynamics of \eqref{model} are much more complicated. If two prey species have the same ecological characteristics (i.e. symmetric apparent competition), we see that the dynamics are rather intricate as shown in Sect. \ref{sbusec4.1} for parameters given in  \eqref{eq4.3}. From this case study, we find that the initial mass of the invasive prey species and the predator's mortality rate are two key factors determining the local and global dynamics of the system (see Table \ref{table4} and bifurcation diagrams plotted in Fig. \ref{fig2}), the success of invasion and the effect of the predator-mediated apparent competition (see numerical simulations illustrated in Sect. \ref{sbusec4.1}). When the predator's mortality rate is low (i.e. the predator is healthy physically), the invasive prey species invades successfully (and hence promotes the predator-mediated apparent competition) only if its initial mass is suitably large. To further wipe out the native prey species via the predator-mediate apparent competition, the initial invasive mass needs to be larger (see numerical simulations in Fig. \ref{fig3}). However, if the predator's mortality rate is moderate  (i.e. the predator is neither very healthy nor poor physically), the large initial mass of the invasive prey species can only ensure the success of the invasion but can not annihilate the native prey species via the predator-mediated apparent competition.  The study of the symmetric apparent competition implies that if both invasive and native prey species are ecologically identical (or similar),  the predator-mediates apparent competition can not be evoked unless a large initial mass of the invasive species is introduced (which is costly however). Once the predator-mediates apparent competition is triggered, the population abundance of the native prey species will be reduced (see Fig. \ref{fig3}). For the asymmetrical apparent competition (i.e. the two prey species have different ecological characteristics), assuming that two prey species have different capture rates (i.e. $\gamma_1\ne \gamma_2$), we find that the invasive prey species will always invade successfully regardless of its initial population abundance if it has a smaller capture rate than the native prey species (see Table \ref{table6} for the case $0<\gamma_2<\gamma_1=1$).   Hence in this case, the initial mass of invasive prey species is no longer important for the promotion of predator-mediated apparent competition which can be triggered easily by introducing a small number of invasive prey species.  If the predator's mortality rate is low (i.e. the predator is healthy physically),  the invasion of a secondary prey species will result in the extinction of the native prey species (Table \ref{table6}). Though we only perform the detailed analysis for the case $0<\gamma_2<\gamma_1=1$,  similar behaviors are numerically found for other asymmetrical parameter regimes as discussed in Sect. \ref{sbusec4.1}. On the contrary, if the invasive prey species has a large capture rate,  it will not be able to promote the predator-mediate apparent competition no matter how large the initial invasive mass is.  This again shows that the capture rate of the prey species and the predator's mortality rate are two major determinants of the predator-mediated apparent competition and its biological consequence for the Holling type II functional response function.
\end{itemize}
In summary, we find that if two prey species employ the same Holling type I response function, whether the invasion is successful and hence promotes the predator-mediated apparent competition is entirely determined by their capture rates (i.e. the rates being captured by the predator).  In contrast, the dynamics with the Holling type II response function are more complicated. First, if two prey species have the same ecological characteristics, then the initial mass of the invasive prey species is the key factor determining the success of the invasion and hence the promotion of the predator-mediated apparent competition. Whereas if two prey species have different ecological characteristics, say different capture rates without loss of generality,  then the success of the invasion (i.e. the promotion of the predator-mediated apparent competition) no longer depends on the initial mass of the invasive prey species, but on the capture rates. In all cases, if the invasion succeeds, whether the native prey species can be annihilated via predator-mediated apparent competition essentially depends on the predator's mortality rate (i.e. the low predator's mortality rate will result in the extinction of the native prey species). All these interesting findings will provide valuable suggestions for decision-makers when introducing alien species to a new ecological system to maintain ecological balance and biodiversity.

Our present works not only pinpoint key factors promoting predator-mediated apparent competition but also show the significant effects of predator-mediated apparent competition on the structure and stability of ecological systems.  Therefore, a comprehensive understanding of the mechanism and underlying dynamics of this indirect interaction is imperative.  This paper only takes a (first) step forward in this direction and  many interesting questions remain open in our present works. Among others, we mention several possible questions for future efforts.
\begin{itemize}[leftmargin=5mm]
    \item We consider the same functional response functions for both prey species, either Holling type I or Holling type II. In reality, the functional response functions for two prey species may be different, such as Holling type I for the native prey species and Holling type II for the invasive one, or vice versa. Then we anticipate that the dynamics might be different from those obtained in this paper. This deserves to be clarified in a future work.
    \item The model considered in this paper does not include the spatial structures, such as the random diffusion and/or directed movement (e.g. prey-taxis cf. \cite{KO1987AN0}), which are indispensable factors to make the model more realistic. This raises a natural question: what are the dynamics of the predator-mediated apparent competition with spatial structure and whether the spatial movement of species will bring significantly different effects?  These interesting questions can serve as a roadmap to study spatial effects on the population dynamics of predator-mediated apparent competition and hence provide insights into the understanding of complex dynamics of ecological systems.  We shall explore this question in the future.
    \item In the model,  the direct (i.e. interference) competition of two prey species is not considered.  If we include the direct competition in the model, the complexity of both qualitative and quantitative analysis will be considerably increased. However, it is still very interesting to explore how the direct competition and indirect interaction (i.e. predator-mediated apparent competition) between the two prey species jointly affect the population dynamics.
\end{itemize}

\section{Proof of the global stability}\label{sec3}

In this section, we aim to prove the global stability of the equilibria of {the} system \eqref{model}. As mentioned earlier, we will primarily focus on proving the global stability of the equilibrium $E_{uv}$ for $\theta\geq L$ and {semi-coexistence}/{coexistence equilibria} for $0<\theta<L$. Before proceeding with the stability analysis, we first establish the {global} well-posedness of {the} system \eqref{model} by proving the following result.
\begin{lem}
    Let $(u_0, v_0, w_0) \in \mathbb{R}_{+}^3$, and let $f_1(u)$ and $f_2(v)$ be given by \eqref{eqh1} or \eqref{eqh2}. Then {the} system \eqref{model} admits a unique {nonnegative} solution, which is bounded for $t{>}0$. Moreover, the solution satisfies
    \begin{align}\label{eq3.1}
        \sup_{t{\geqslant}0} u(t)\leq M_1,\quad
        \sup_{t{\geqslant}0} v(t)\leq M_2,\quad
        \sup_{t{\geqslant}0} w(t)\leq M_3,
    \end{align}
    and
    \begin{align}\label{eq3.2}
        \limsup _{t \rightarrow \infty} u(t) \leq K_1,\quad
        \limsup _{t \rightarrow \infty} v(t) \leq K_2,\quad
        \limsup _{t \rightarrow \infty} w(t) \leq \frac{(1+\theta)^2}{4\theta}\xkh{\beta_1K_1+\beta_2K_2},
    \end{align}
    where the constants $M_i$, $i=1,2,3$, are given by
    \begin{align*}
        \begin{cases}
            \medskip
            M_1:=\max\dkh{u_0,K_1},\quad
            M_2:=\max\dkh{v_0,K_2},\\
            M_3:=\max\dkh{\beta_1u_0+\beta_2v_0+w_0,\frac{(1+\theta)^2}{4\theta}\xkh{\beta_1K_1+\beta_2K_2}}.
        \end{cases}
    \end{align*}
\end{lem}
\begin{proof}
    Since the vector field, defined by the terms on the right-hand side of {the} system \eqref{model}, is smooth in $\mathbb{R}_{+}^3$, the existence theory of ordinary differential equations (cf. \cite[Theorem 4.18]{LR20140}) guarantees that {the} system \eqref{model} admits a unique maximal solution with a maximal time $T_{max}\in(0,\infty]$. {By the first equation of \eqref{model}, we have
    \begin{align*}
        u(t)=u_0e^{\int_0^t(1-\frac {u(s)}{K_1}-\frac{w(s)f_1(s)}{u(s)})}ds\geq0\quad\text{for all }t\in(0,T_{max}).
    \end{align*}
    We can similarly obtain $v(t),w(t)\geq0$ for all $t\in(0,T_{max})$. Then the first equation of \eqref{model} shows that} $u_t\leq u(1-u / K_{1})$, which along with the comparison principle gives $u(t)\leq \max\dkh{u_0,K_1}=M_1$ for all $t\in(0,T_{max})$. {Similarly, it} holds that $v(t)\leq  \max\dkh{v_0,K_2}=M_2$ for all $t\in(0,T_{max})$. Let $z(t):=\beta_1u(t)+\beta_2v(t)+w(t)$, then we have from \eqref{model} and Young's inequality that
    \begin{align*}
        z_t\deyu ~\beta_1u\left(1-\frac u {K_{1}}\right)+\beta_2 v\left(1-\frac v {K_{2}}\right)-\theta\xkh{z-\beta_1u -\beta_2v }\non
        \deyu -\theta z+\beta_1\xkh{\xkh{1+\theta}u-\frac{u^2}{K_1}}+\beta_2\xkh{\xkh{1+\theta}v-\frac{v^2}{K_2}}\non
        \xiyu -\theta z+\frac{(1+\theta)^2}4\xkh{\beta_1K_1+\beta_2K_2}\quad \text{for all } t\in(0,T_{max}).
    \end{align*}
    By {the} comparison {principle}, we get $w(t)\leq z(t)\leq M_3=\max\{\beta_1u_0+\beta_2v_0+w_0,\frac{(1+\theta)^2}{4\theta}\xkh{\beta_1K_1+\beta_2K_2}\}$ for all $t\in(0,T_{max})$. Therefore the solution is bounded and hence $T_{max}=\infty$. Given the above analysis, \eqref{eq3.1} and \eqref{eq3.2} follow immediately. {The proof is completed.}
\end{proof}

Now we consider the global stability of the equilibria of {the} system \eqref{model}. Before proceeding, for $t>0$ and a given equilibrium $E_s=(u_s,v_s,w_s)$, we let
\begin{align}\label{eq3.3}
    \mathcal E(t;E_s):=
    \Gamma_1 \xkh{u-u_s-u_s\ln\frac u{u_s}}
    +\Gamma_2 \xkh{v-v_s-v_s\ln\frac v{v_s}}
    + \xkh{w-w_s-w_s\ln\frac w{w_s}},
\end{align}
where the constants $\Gamma_1$ and $\Gamma_2$ are given by
\begin{align}\label{eq3.4}
    \Gamma_i:=
    \begin{cases}
        \medskip
        \frac{\beta_i f_i(u_s)}{\alpha_i u_s}
        =\beta_i,\quad &\text{if }\eqref{eqh1}\text{ holds},\\
        \frac{\beta_i f_i(u_s)}{\gamma_i u_s}
        =\frac{\beta_i}{1+\gamma_ih_iu_s},\quad &\text{if }\eqref{eqh2}\text{ holds},
    \end{cases}
    \qquad i=1,2.
\end{align}
Then we {prove} the global {stability} of the equilibria.
\subsection{Global stability for $\theta\geq L $}
The first result asserts that the equilibrium $E_{uv}$ is globally asymptotically stable if $\theta\geq L$.
\begin{lem} \label{lem3.2}
    Let $\theta\geq L$, and let $f_1(u)$ and $f_2(v)$  be given by \eqref{eqh1} or \eqref{eqh2}. Then the equilibrium $E_{uv}$ is globally asymptotically stable.
\end{lem}
\begin{proof}
    {We first recall that $u,v,w\geq0$ for all $t\geq0$.} Let $E_s=E_{uv}=(K_1,K_2,0)$ in \eqref{eq3.3} and \eqref{eq3.4}. Then
    \begin{equation}\nonumber
        \medskip
        \mathcal E(t;E_{uv})=\Gamma_1\xkh{u-K_1-K_1\ln\frac u{K_1}}
        +\Gamma_2 \xkh{v-K_2-K_2\ln\frac v{K_2}}+w,
    \end{equation}
    and
    \begin{align}\label{eq3.5}
        \mathcal E' (t;E_{uv})
        \deyu ~\Gamma_1 \xkh{1-\frac u{K_1}-\frac{wf_1(u)}u}\xkh{u-K_1}+\Gamma_2 \xkh{1-\frac v{K_2}-\frac{wf_2(v)}v}\xkh{v-K_2}\non
        & +w\xkh{\beta_1 f_1(u)+\beta_2 f_2(v)-\theta}.
    \end{align}
    We claim that
    \begin{align}\label{eq3.6}
        \left\{
        \begin{array}{l}
            \medskip
            \mathcal E(t;E_{uv})>0\text{ for all }(u,v,w)\neq E_{uv},\\
            \mathcal E'(t;E_{uv})\leq0{,\text{ where }
            ``="}\text{ {holds} if and only if }
            (u,v,w)=E_{uv}.
        \end{array}
        \right.
    \end{align}

    Indeed, for any given $c_0>0$, the function $\phi_1(s):=s-c_0-c_0\ln \frac s{c_0}$ for $s>0$ satisfies $\phi'_1(s)=1-\frac{c_0}s$ {and $\phi''_1(s)=\frac{c_0}{s^2}>0$}, which implies that $\phi_1(s)\geq\phi_1(c_0)=0$ and $\phi_1(s)=0$ if and only if $s=c_0$. Therefore, the first conclusion in \eqref{eq3.6} {follows}. Moreover, if \eqref{eqh1} holds, then \eqref{eq3.4} gives $\Gamma_1=\beta_1$ and $\beta_1 \xkh{f_1(u)-f_1(K_1)}=\Gamma_1 \frac{f_1(u)}u\xkh{u-K_1}$. If \eqref{eqh2} holds, then \eqref{eq3.4} gives $\Gamma_1=\frac{\beta_1}{1+\gamma_1h_1K_1}$ and
    \begin{align}\label{eq3.7}
        \beta_1 \xkh{f_1(u)-f_1(K_1)}
    = \frac{\gamma_1\beta_1(u-K_1)}{(1+\gamma_1h_1K_1)(1+\gamma_1h_1 u)}
    =\Gamma_1 \frac{f_1(u)}u\xkh{u-K_1}.
    \end{align}
    Similarly, we have
    \begin{align}\label{eq3.8}
        \beta_2 \xkh{f_2(v)-f_2(K_2)}
    =\Gamma_2 \frac{f_2(v)}v\xkh{v-K_2}.
    \end{align}
    Using \eqref{eqh1}, \eqref{eqh2}, \eqref{eq3.3}, \eqref{eq3.7}, \eqref{eq3.8} and $\theta\geq L$, we have
    \begin{align*}
        w\xkh{\beta_1 f_1(u)+\beta_2 f_2(v)-\theta}
        \xiyu~ \beta_1  w\xkh{f_1(u)-f_1(K_1)}+\beta_2 w \xkh{f_2(v)-f_2(K_2)}\non
        \deyu ~\Gamma_1 \frac{wf_1(u)}u(u-K_1)+\Gamma_2 \frac{wf_2(v)}v(v-K_2),
    \end{align*}
    which along with \eqref{eq3.5} yields
    \begin{align*}
        \mathcal E'(t;E_{uv})\leq
         -\frac{\Gamma_1}{K_1} (u-K_1)^2-\frac{\Gamma_2}{K_2} (v-K_2)^2.
    \end{align*}
    The above inequality indicates that {$\mathcal E'(t;E_{uv})\leq 0$, where ``="  holds in the case of $(u,v)= (K_1,K_2)$. Note that if $(u,v)=(K_1,K_2)$,  the first equation of \eqref{model} becomes $0=wf_1(K_1)$, which implies $w=0$ due to $f_1(K_1)>0$. Therefore, $\mathcal E'(t;E_{uv})<0$ if $(u,v,w)\neq E_{uv}$.} Clearly, \eqref{eq3.5} implies $\mathcal E'(t;E_{uv})=0$ for $(u,v,w)= E_{uv}$. Hence \eqref{eq3.6} is proved. With \eqref{eq3.6} and LaSalle's invariant principle (cf. \cite[Theorem 3]{L1960T0}), the proof is completed.
\end{proof}

In what follows we assume $\theta\in(0, L)$ and consider two types of functional response functions separately.

\subsection{Global stability for $\theta\in(0,L)$ and Holling type I $(\ref{eqh1})$}

We next show that the unique {coexistence equilibrium} $P_*$ is globally asymptotically stable as long as it exists.

\begin{lem}[Global stability of $P_*$]\label{lem3.3}
    Let \eqref{eqh1} hold and $\theta\in(\theta_0,L )$. Then the unique {coexistence equilibrium} $P_*=\xkh{u_*,v_*,w_*}$ of \eqref{model} is globally asymptotically stable.
\end{lem}
\begin{proof}
    Let $E_s=P_*=\xkh{u_*,v_*,w_*}$ in \eqref{eq3.3} and \eqref{eq3.4}{. Then} \eqref{eq3.4} implies $\Gamma_i=\beta_i$, $i=1,2$. Using \eqref{model}, \eqref{eqh1}, \eqref{eq3.3}, \eqref{eq3.4} and the fact that
    \begin{align}\nonumber
        \theta=\beta_1f_1(u_*)+\beta_2f_2(v_*)=\alpha_1\beta_1 u_*+\alpha_2\beta_2 v_*,\quad
            1=\frac{u_*}{K_1}+\alpha_1w_*
             =\frac{v_*}{K_2}+\alpha_2w_*,
    \end{align}
    we have
    \begin{align*}
        \mathcal E'(t;P_*)
        \deyu~ \beta_1 \xkh{1-\frac u{K_1}-\alpha_1 w}\xkh{u-u_*}+\beta_2 \xkh{1-\frac v{K_2}-\alpha_2 w}\xkh{v-v_*} \non
        &+ (\alpha_1\beta_1 u+\alpha_2\beta_2 v-\theta)(w-w_*)\non
        \deyu~\beta_1 \xkh{1-\frac u{K_1}-\alpha_1w_*}\xkh{u-u_*}+\beta_2 \xkh{1-\frac v{K_2}-\alpha_2w_*}\xkh{v-v_*}\non
        \deyu-\frac{\beta_1}{K_1} \xkh{u-u_* }^2-\frac{\beta_2}{K_2} \xkh{v-v_* }^2.
    \end{align*}
    {Hence $\mathcal E'(t;P_*)\leq 0$, where ``=" possibly holds in the case of $(u,v)= (u_*,v_*)$. Note that if $(u,v)= (u_*,v_*)$, then $w=w_*$ since the system \eqref{model} admits the unique coexistence equilibrium $P_*$ for $\theta\in(\theta_0,L )$. Therefore, $\mathcal E'(t;P_*)<0$ if $(u,v,w)\neq E_{uv}$. If $(u,v,w)=P_*$, then \eqref{eq3.3} obviously shows that $\mathcal E(t;P_*)=0$ for all $t>0$, which implies that $\mathcal E'(t;P_*)=0$ for all $t>0$. We obtain
    $$\mathcal E'(t;E_{uv})\leq0,\text{ where }
    \text{``='' holds if and only if }
    (u,v,w)=E_{uv}.$$}
    {Moreover, the} same arguments as in the proof of Lemma \ref{lem3.2} yield that $\mathcal E(t;P_*)>0$ for $(u,v,w)\neq P_*$. Then the proof is completed by an application of LaSalle's invariant principle.
\end{proof}

Note that $\theta_0=0$ if and only if $\alpha_1=\alpha_2$. If $\alpha_1=\alpha_2$, then Lemma \ref{lem3.2} and Lemma \ref{lem3.3} imply that for any $\theta>0$, either $E_{uv}$ or $P_*$ is globally asymptotically stable. We next consider the case $\alpha_1\neq\alpha_2$, which implies $\theta_0>0$. Then in view of Table \ref{table1}, the {semi-coexistence} equilibria $P_1$ and $P_2$ exist for $\theta\in(0,\theta_0]$ since $0<\theta_0<\min\dkh{L_1,L_2}$. The following result gives {the} global {stability} of $P_1$ and $P_2$.

\begin{lem}[Global stability of $P_1$ and $P_2$]\label{lem3.4}
    Let \eqref{eqh1} hold, $\alpha_2>\alpha_1$ (resp. $\alpha_1>\alpha_2$) and $\theta\in(0,\theta_0]$. Then the {semi-coexistence} equilibrium $P_1$ (resp. $P_2$) of \eqref{model} is globally asymptotically stable.
\end{lem}
\begin{proof}
    Without loss of generality, we only prove {the global stability} for $P_1=\xkh{u_{P_1},0,w_{P_1}}=(\frac\theta{\alpha_1\beta_1}, 0,\frac{L_1-\theta}{\alpha_1L_1})$ in the case of $\alpha_2>\alpha_1$, and the proof {for $P_2$} in the case of $\alpha_1>\alpha_2$ is similar. Let $E_s=P_1$ in \eqref{eq3.3} and \eqref{eq3.4}, then \eqref{eq3.4} implies $\Gamma_i=\beta_i$, $i=1,2$. Clearly, $0<\theta_0=(1-\frac{\alpha_1}{\alpha_2})L_1<L_1$, and hence
    $$\alpha_2 {w_{P_1}}=\frac{\alpha_2}{\alpha_1}\xkh{1-\frac\theta{L_1}}
    \geq \frac{\alpha_2}{\alpha_1}\xkh{1-\frac{\theta_0}{L_1}}
    = 1,$$
    which alongside \eqref{model}, \eqref{eqh1}, $\theta=\alpha_1\beta_1 u_{P_1}$ and $\alpha_1 w_{P_1}=1-\frac {u_{P_1}}{K_1}$ implies that
    \begin{align*}
        \mathcal E'(t;P_1)
        \deyu ~\beta_1 \xkh{1-\frac u{K_1}-\alpha_1w}\xkh{u-u_{P_1}}+\beta_2 \xkh{1-\frac v{K_2}-\alpha_2 w}v\non
        &  +\xkh{\beta_1 f_1(u)+\beta_2 f_2(v)-\theta}\xkh{w-w_{P_1}}\non
        \deyu~ \beta_1 \xkh{1-\frac u{K_1}-\alpha_1{w_{P_1}}}\xkh{u-u_{P_1}}+\beta_2 \xkh{1-\frac v{K_2}-\alpha_2 {w_{P_1}}}v\non
        \xiyu -\frac{\beta_1}{K_1}\left(u-u_{P_1}\right)^2
        -\frac{\beta_2}{K_2}v^2.
    \end{align*}
    The rest of the proof is similar to that of Lemma \ref{lem3.2}, and we omit it for brevity.
\end{proof}
\subsection{Global stability for $\theta\in(0,L )$ and Holling type II $(\ref{eqh2})$}

We now consider the case of Holling type II \eqref{eqh2}. We first give the global stability of {semi-coexistence} equilibria $Q_1$ and $Q_2$.

\begin{lem}[Global stability of $Q_1$ and $Q_2$]\label{lem3.5}
    Let \eqref{eqh2} hold and $\theta\in(0,L_1)$ (resp. $\theta\in(0,L_2)$). Then the {semi-coexistence} equilibrium $Q_1$ (resp. $Q_2$) of \eqref{model} is globally asymptotically stable if \eqref{eq2.1} (resp. \eqref{eq2.2}) holds.
\end{lem}
\begin{proof}
    Without loss of generality, we only prove {the global stability} for $Q_1=\xkh{u_{Q_1},0,w_{Q_1}}$, and the case for $Q_2$ can be proved similarly. Let $E_s=Q_1=\xkh{u_{Q_1},0,w_{Q_1}}$ in \eqref{eq3.3} and \eqref{eq3.4}{. Then} \eqref{eq3.4} implies
    \begin{align}\label{eq3.9}
        \Gamma_1=\frac{\beta_1}{1+\gamma_1h_1u_{Q_1}}
        \quad\text{and}\quad
        \Gamma_2=\beta_2
    \end{align}
    If $v_0\leq K_2$, then \eqref{eq3.1} implies $v(t)\leq K_2$ for all $t\geq 0$. This along with \eqref{eq2.1} and the fact $\frac{f_2(s)}{s}$ {decreases} for $s\geq0$ indicates {that}
    $\frac{w_{Q_1}f_2(v)}{v}\geq \frac{w_{Q_1}f_2(K_2)}{K_2}\geq1$ for all $t\geq0$. Similarly, if $v_0> K_2$, then \eqref{eq2.1} implies $\frac{K_2}{f_2(K_2)}< w_{Q_1}${. Hence} \eqref{eq3.2} {yields} $T_1>0$ such that $\frac{w_{Q_1}f_2(v)}{v}>1$ for all $t\geq T_1$. Therefore, for $v_0\geq0$, it  holds that
    \begin{align}\label{eq3.10}
        \frac{w_{Q_1}f_2(v)}{v}\geq 1\quad\text{for all }t\geq T_1.
    \end{align}
    Using $\theta=\beta_1f_1(u_{Q_1})$, \eqref{eqh2}, \eqref{eq3.9} and \eqref{eq3.10}, we have
    \begin{align}\nonumber
        \beta_1 f_1(u)-\theta
        =\beta_1\xkh{f_1(u)-f_1(u_{Q_1})}
        =\beta_1 \frac{\gamma_1\xkh{u-u_{Q_1}}}{\xkh{1+\gamma_1h_1u}\xkh{1+\gamma_1h_1u_{Q_1}}}=\Gamma_1\frac{f_1(u)}u\xkh{u-u_{Q_1}}.
    \end{align}
    {Consequently,}
    \begin{align}
        \mathcal E'(t;Q_1)
        \deyu ~\Gamma_1 \xkh{1-\frac u{K_1}-\frac{wf_1(u)}u}\xkh{u-u_{Q_1}}+\beta_2 \xkh{1-\frac v{K_2}-\frac{wf_2(v)}v}v\non
        &  +\xkh{\beta_1 f_1(u)+\beta_2 f_2(v)-\theta}\xkh{w-w_{Q_1}}\non
        \deyu~ \Gamma_1 \xkh{1-\frac u{K_1}-\frac{w_{Q_1}f_1(u)}{u}}\xkh{u-u_{Q_1}}
        +\beta_2 \xkh{1-\frac v{K_2}-\frac{w_{Q_1}f_2(v)}{v}}v\non
        \deyu -\frac{\Gamma_1h_1\gamma_1(\lambda_1+u_{Q_1}-K_1+u)}{K_1(1+h_1\gamma_1u)}(u-u_{Q_1})^2
        -\frac{\beta_2}{K_2}v^2
        +\beta_2 \xkh{1-\frac{w_{Q_1}f_2(v)}{v}}v\non
        \xiyu -\frac{\Gamma_1h_1}{K_1}f_1(u)(u-u_{Q_1})^2 -\frac{\beta_2}{K_2}v^2\quad\text{for all } t\geq T_1.\nonumber
    \end{align}
    {Then the global stability of $Q_1$ follows from the Lyapunov function method and LaSalle's invariant principle, similar as in the proof of Lemma \ref{lem3.2}.}
\end{proof}
We next prove the global stability of the {coexistence equilibrium} $Q_*$.
\begin{lem}[Global stability of $Q_*$]\label{lem3.6}
    Let \eqref{eqh2} hold, $\theta\in(0,L)$ and $Q_*$ be {a coexistence equilibrium} of \eqref{model}. Then $Q_*$ is globally asymptotically stable if \eqref{eq2.3} holds.
\end{lem}
\begin{proof}
    Let $E_s=Q_*=\xkh{u_*,v_*,w_*}$ in \eqref{eq3.3} and \eqref{eq3.4}{. Then} \eqref{eq3.4} {gives}
    \begin{align}\nonumber
        \Gamma_1= \frac{\beta_1 }{1+\gamma_1h_1u_*},\quad
        \Gamma_2= \frac{\beta_2 }{1+\gamma_2h_2v_*}.
    \end{align}
    Using \eqref{model}, \eqref{eqh1} and $\theta=\beta_1 f_1(u_*)+\beta_2 f_2(v_*)$, we obtain
    \begin{align}\label{eq3.11}
         \mathcal E' (t;Q_*)
         \deyu~ \Gamma_1 \xkh{1-\frac u{K_1}-\frac{wf_1(u)}{u}}\xkh{u-u_*}+\Gamma_2 \xkh{1-\frac v{K_2}-\frac{wf_2(v)}{v}}\xkh{v-v_*} \non
         & +  \beta_1\xkh{ f_1(u)-f_1(u_*)}\xkh{w-w_*}+\beta_2\xkh{ f_2(v)-f_2(v_*)}\xkh{w-w_*}.
    \end{align}
   Similar as in driving \eqref{eq3.7} and \eqref{eq3.8}, we have
    \begin{align}\label{eq3.12}
        \beta_1\xkh{ f_1(u)-f_1(u_*)} =\Gamma_1\frac{f_1(u)}{u}(u-u_*)
        \quad\text{and}\quad
        \beta_2\xkh{ f_2(v)-f_2(v_*)}=\Gamma_2\frac{f_2(v)}{v}(v-v_*).
    \end{align}
    Substituting \eqref{eq3.12} into \eqref{eq3.11} and using
    $w_*=\frac{u_*}{f_1(u_*)}(1-\frac{u_*}{K_1})=\frac{v_*}{f_2(v_*)}(1-\frac{v_*}{K_2})$ yields
    \begin{align*}
        \mathcal E'(t;Q_*)\deyu~  \Gamma_1 \xkh{1-\frac u{K_1}-\frac{w_*f_1(u)}{u}}\xkh{u-u_*}+\Gamma_2 \xkh{1-\frac v{K_2}-\frac{w_*f_2(v)}{v}}\xkh{v-v_*}\non
        \deyu-\frac{\Gamma_1h_1\gamma_1(\lambda_1+u_*-K_1+u)}{K_1(1+h_1\gamma_1u)}(u-u_*)^2- \frac{\Gamma_2 h_2\gamma_2(\lambda_2+v_*-K_2+v)}{K_2(1+h_2\gamma_2v)}(v-v_*)^2\non
        \xiyu -\frac{\Gamma_1h_1}{K_1}f_1(u)(u-u_*)^2- \frac{\Gamma_2h_2}{K_2}f_2(v)(v-v_*)^2.
    \end{align*}
    {Similar arguments with the Lyapunov function method alongside LaSalle's invariant principle as above complete the proof.}
\end{proof}
{
\noindent\textbf{Proof of Theorem \ref{thm2.1}.} In view of Lemmas \ref{lem3.2}, \ref{lem3.3} and \ref{lem3.4}, Theorem \ref{thm2.1} is proved. $\hfill\square$}
\vskip3mm
\noindent\textbf{Proof of Theorem \ref{thm2.2}.} {With the results from Lemmas \ref{lem3.2}, \ref{lem3.5} and \ref{lem3.6}, Theorem \ref{thm2.2} is obtained.} $\hfill\square$

\section*{Acknowledgments}
The research of Y. Lou is partially supported by the NSF of China No. 12261160366 and No. 12250710674). The research of W. Tao is partially supported by PolyU Postdoc Matching Fund Scheme Project ID P0030816/B-Q75G, 1-W15F and 1-YXBT, and the NSF of China (No. 12201082). The research of Z.-A. Wang was partially supported by the NSFC/RGC Joint Research Scheme sponsored by the Research Grants Council of Hong Kong and the National Natural Science Foundation of China (Project No. \mbox{$\mathrm{N}_{-}$PolyU509/22}), and PolyU Postdoc Matching Fund Scheme Project ID P0034904 (Primary Work Programme W15F).
\vspace{0.5cm}

\noindent {\bf Declaration}
\vspace{0.2cm}

\noindent{\bf Data availability}. The authors declare that the manuscript has no associated data.

\noindent{\bf Conflict of interest}. The authors have no conflicts of interest to declare that are relevant to the content of
this article.

\section*{Appendix A. Proof for Remark \ref{rem2.2}}
\setcounter{equation}{0}
\setcounter{subsection}{0}
\setcounter{thm}{0}
\renewcommand{\theequation}{A\arabic{equation}}
\renewcommand{\thesubsection}{A.\arabic{subsection}}
\renewcommand{\thethm}{A.\arabic{thm}}
This appendix is dedicated to proving the conclusion stated in Remark \ref{rem2.2}.
\begin{proof}[Proof for Remark \ref{rem2.2}]
    We first prove
    \begin{align}\label{A1}
        \Lambda_1\cap \Lambda_2=\emptyset.
    \end{align}
    Note that $w_{Q_1}$ strictly increases with respect to $K_1>0$ since
    \[
      \frac{d w_{Q_1}}{d K_1}
      =\frac{u_{Q_1}^2}{f_1(u_{Q_1}) K_1^2}
      =\frac{\beta_1 u_{Q_1}^2}{\theta K_1^2}>0.
    \]
    Let $(K_1, K_2)\in\Lambda_1$. If $\gamma_1\geq\gamma_2$, then the first condition in $\Lambda_1$ implies that
    \begin{align}\nonumber
        w_{Q_1}\leq w_{Q_1}\big|_{K_1=\lambda_1+u_{Q_1}}=\frac1{\gamma_1}\leq \frac1{\gamma_2}<\frac{1+\gamma_2h_2K_2}{\gamma_2}=\frac{K_2}{f_2(K_2)},
    \end{align}
    which contradicts $(K_1, K_2)\in\Lambda_1$. Therefore, $\Lambda_1=\emptyset$ in the case of $\gamma_1\geq\gamma_2$. Similarly one can show that  $\Lambda_2=\emptyset$ in the case of $\gamma_1\leq\gamma_2$. The claim \eqref{A1} is proved.

    Without loss of generality, we next assume  $\gamma_1\geq\gamma_2$, then $\Lambda_1=\emptyset$. It remains to prove that
    \begin{align}\label{A2}
        \Lambda_2\cap\Lambda_*=\emptyset.
    \end{align}
    Assuming that there exists a pair $(K_1, K_2)\in\Lambda_2\cap\Lambda_*$, we shall derive a contradiction. Using
    \begin{align}\nonumber
        \theta=\beta_2f_2(v_{Q_2})=\beta_{1} f_{1}(u_*)+\beta_{2} f_{2}(v_*)
    \end{align}
    and the fact that $f_i(s)$ and $\frac s{f_i(s)}$ ($i=1,2$) strictly increase with respect to $s\geq0$, we have
    \begin{align}\label{A3}
        v_*<v_{Q_2}<K_2
        \quad\text{and}\quad
        w_*=\frac{u_*}{f_1(u_*)}\xkh{1-\frac{u_*}{K_1}}<\frac{u_*}{f_1(u_*)}< \frac{K_1}{f_1(K_1)}.
    \end{align}
    By the second condition in $\Lambda_2$ and \eqref{A3}, we have $w_*<w_{Q_2}$, which means that
    \begin{align}\label{A4}
        w_*=\varphi_0(v_*)<\varphi_0(v_{Q_2})=w_{Q_2},
    \end{align}
    where
    \begin{align}\label{A5}
        \varphi_0(s):=&~\frac{s}{f_2(s)}\xkh{1-\frac{s}{K_2}}
        =\frac{(K_2-s) (1+h_2\gamma_2 s)}{K_2\gamma_2}\non
        =&-\frac{h_2}{K_2}\xkh{s-\frac{K_2-\lambda_2}2}^2+\frac{h_2(K_2-\lambda_2)^2}{4K_2}+\frac1{\gamma_2},\quad s\in[0,K_2].
    \end{align}
    The combination of the first equation of \eqref{A3}, \eqref{A4} and \eqref{A5} implies that
    \begin{align}\label{A6}
        v_*<\frac{K_2-\lambda_2}2
        \quad\text{and}\quad
        v_*+v_{Q_2}<K_2-\lambda_2.
    \end{align}
    Starting from the first condition in $\Lambda_2$, the second condition in $\Lambda_*$, and the second inequality in \eqref{A6}, we obtain $2K_2<2\lambda_2+v_*+v_{Q_2}<2\lambda_2+K_2-\lambda_2$, which simplifies to $K_2\leq \lambda_2$. Therefore, the first inequality in \eqref{A6} indicates $v_*\leq0$ which is absurd. This proves \eqref{A2} and hence proves that $\Lambda_1$, $\Lambda_2$ and ${\Lambda_*}$ are mutually disjoint.
\end{proof}

\section*{Appendix B. }
\setcounter{equation}{0}
\setcounter{subsection}{0}
\setcounter{thm}{0}
\renewcommand{\theequation}{B\arabic{equation}}
\renewcommand{\thesubsection}{B.\arabic{subsection}}
\renewcommand{\thethm}{B.\arabic{thm}}
This appendix is devoted to proving that the rescaled system \eqref{eq4.8} with \eqref{eg2} has at most one coexistence equilibrium $Q_*=( u_*, v_*, w_*)$ (see \eqref{eq4.9}), which exists if and only if $\theta\in(\Theta_1, L)$, where $\Theta_1$ and $L$ are given by \eqref{eq4.11}. Remark \ref{rem2.1} implies $\theta\in(0,L)$ is a necessary condition for the existence of $Q_*$. Therefore, we shall consider $\theta\in(0, L)$ below. Within this appendix, we shall use the notations defined in \eqref{eq4.11}. For clarity, we also introduce the following notations.
\begin{itemize}
    \item For $\gamma_2\in(0,1)$, $b>0$ and $\theta\in(0,L)$, let
    \begin{align}\label{B1}
        v_M:=\frac{\theta }{\gamma_2 (b-\theta )}
    \end{align}
    be a positive constant (note that \eqref{eq4.11} implies $\theta<b$), then $v_M$ strictly increases in $\theta\in(0,L)$ and
    \begin{align}\label{B2}
        v_M
        \begin{cases}
            \medskip
            <1,\quad&\text{if }\theta\in(0,L_2),\\
            \geq1,\quad&\text{if }\theta\in[L_2,L).
        \end{cases}
    \end{align}
    \item It is straightforward to check that either of the equations
    \begin{align}\nonumber
        \begin{cases}
            \medskip
            8 s^3+7 s^2-8 s+1=0,\\
            24 s^3-13 s^2-6 s+3=0,
        \end{cases}
        \quad s\in\mathbb{R},
    \end{align}
    has two positive (real) roots and one negative (real) root. Denote the two positive roots of the first equation by $\eta_1$ and $\eta_3$ with $\eta_1<\eta_3$, and the largest root of the second equation by $\eta_4$. Let $\eta_2=\frac{4 \sqrt{7}-7}{9}$. Then
    \begin{align}\nonumber
        (\eta_1,\eta_2,\eta_3,\eta_4)\approx
        (0.1471,0.3981,0.5429,0.6195).
    \end{align}
    \item For $b>0$, define the functions
    \begin{align*}
        \xi_1(\gamma_2):=&\frac{b \left(4 \gamma_2^2-5 \gamma_2+1\right)}{1-3 \gamma_2},\quad \gamma_2\in(\eta_1,\frac14)\cup(\eta_3,1),\\
        \xi_2(\gamma_2):=&\frac{b}{3} (4-\gamma_2-\sqrt{\gamma_2^2+\gamma_2+1}),\quad \gamma_2\in(\eta_2,1).
    \end{align*}
    It holds that $\xi_1(\gamma_2)$ strictly decreases in each connected domain with $\xi_1(\gamma_2)\in(0,L)$,
    \begin{align}\nonumber
        \lim_{\gamma_2\rightarrow \eta_1}\xi_1(\gamma_2)=
    \lim_{\gamma_2\rightarrow \eta_3}\xi_1(\gamma_2)=L
    \quad\text{and}\quad
    \lim_{\gamma_2\rightarrow \frac14}\xi_1(\gamma_2)=
    \lim_{\gamma_2\rightarrow 1}\xi_1(\gamma_2)=0.
    \end{align}
    The function $\xi_2(\gamma_2)$ strictly decreases in $(\eta_2,1)$ with $\xi_2(\gamma_2)\in(m_b,L)$,
    $$\lim_{\gamma_2\rightarrow \eta_2}\xi_2(\gamma_2)=L
    \quad\text{and}\quad
    \lim_{\gamma_2\rightarrow 1}\xi_2(\gamma_2)=m_b,$$
    where $m_b:=(1-\frac1{\sqrt3})b$. Moreover, for $\gamma_2\in (\eta_3,1)$,
    \begin{align}\nonumber
        \begin{cases}
            \medskip
            \xi_1(\gamma_2)>\xi_2(\gamma_2),\quad &\text{if }\gamma_2\in(\eta_3,\eta_4),\\
            \medskip
            \xi_1(\eta_4)=\xi_2(\eta_4)\approx 0.6550,\\
            \xi_1(\gamma_2)<\xi_2(\gamma_2),\quad &\text{if }\gamma_2\in(\eta_4,1).
        \end{cases}
    \end{align}
    The graphs of three functions $\frac{\xi_1(\gamma_2)}b$, $\frac{\xi_2(\gamma_2)}b$ and $\frac Lb=\frac{1+3\gamma_2}{2(1+\gamma_2)}$ are shown in Fig. \ref{fig9}(a)
    \item For $b>0$, $\gamma_2\in(0,1)$ and $0<\theta<L$, let
    \begin{align}\label{B3}
        G(s):=\sum_{k=0}^4 D_k s^k,\quad s\geq0,
    \end{align}
    where the coefficients are given by
    \begin{align}\label{B4}
       \begin{cases}
           \medskip
           D_4=\gamma_2^3 ( 2 b-\theta)^2,\\
           \medskip
           D_3=\gamma_2^2 (2 b-\theta ) (2 b (2-\gamma_2)-\theta  (3-\gamma_2)),\\
           \medskip
           D_2=(\gamma_2-1) \gamma_2 \left(b^2 (3 \gamma_2-5)-2 b \theta  (\gamma_2-4)-3 \theta ^2\right),\\
           \medskip
           D_1=b^2 \left(4 \gamma_2^2-5 \gamma_2+1\right)-2 b \theta  \left(2 \gamma_2^2-4 \gamma_2+1\right)+\theta ^2 (1-3 \gamma_2),\\
           D_0=b^2 (\gamma_2-1)-2 b \theta  (\gamma_2-1)-\theta ^2.
       \end{cases}
    \end{align}
\end{itemize}

\begin{figure}[!ht] \centering
    \includegraphics
    [width=1\textwidth,trim=150 10 125 10,clip]{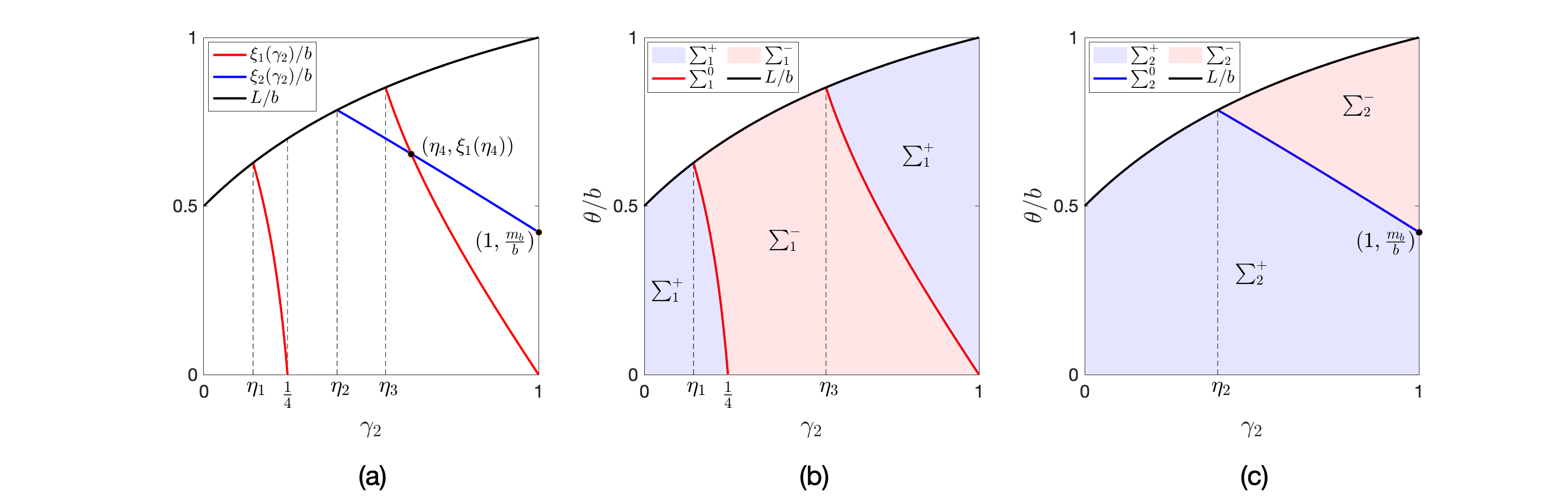}
    \caption{\small The graphs of three functions $\frac{\xi_1(\gamma_2)}b$, $\frac{\xi_2(\gamma_2)}b$ and $\frac Lb$ versus $\gamma_2\in(0,1)$ are shown in (a). The signs of $D_1$ and $D_2$ in the $\gamma_2$-$\theta/b$ plane within $(\gamma_2,\theta)\in(0,1)\times(0, L)$ are shown in (b) and (c), respectively.}
    \label{fig9}
\end{figure}

By elementary analysis (omitted for brevity), we have the following result concerning the signs of the coefficients given in \eqref{B4}.

\begin{prop}\label{propB.1}
    Let $b>0$, $\gamma_2\in(0,1)$ and $\theta\in(0,L)$. Then $D_4,D_3>0$, $D_0<0$,  and
    \begin{align*}
        D_1
        \begin{cases}
            \medskip
            <0,\quad &\text{if }
            (\gamma_2,\theta)
            \in
            \sum_1^-:=(\eta_1,\frac14]\times (\xi_1,L)\cup
            (\frac14,\eta_3]\times (0,L)\cup
            (\eta_3,1)\times (0,\xi_1),\\
            \medskip
            =0,\quad &\text{if }(\gamma_2,\theta)\in
            \sum_1^0:=(\eta_1,\frac14)\times\{\xi_1(\gamma_2)\}\cup(\eta_3,1)\times\{\xi_1(\gamma_2)\},\\
            >0,\quad &\text{if }
            (\gamma_2,\theta)
            \in
            \sum_1^+:=(0,\eta_1]\times (0,L)
            \cup (\eta_1,\frac14)\times (0,\xi_1)\cup
            (\eta_3,1)\times (\xi_1,L),
        \end{cases}
    \end{align*}
    and
    \begin{align*}
        D_2
        \begin{cases}
            \medskip
            <0,\quad &\text{if }
            (\gamma_2,\theta)
            \in\sum_2^-:=(\eta_2,1)\times (\xi_2,L),\\
            \medskip
            =0,\quad &\text{if }(\gamma_2,\theta)
            \in
            \sum_2^0:=(\eta_2,1)\times\{\xi_2(\gamma_2)\},\\
            >0,\quad &\text{if }
            (\gamma_2,\theta)
            \in
            \sum_2^+:=(0,\eta_2]\times (0,L)
            \cup (\eta_2,1)\times (0,\xi_2).
        \end{cases}
    \end{align*}
\end{prop}

Proposition \ref{propB.1} provides a geometric illustration for the signs of $D_1$ and $D_2$ in the $\gamma_2$-$\theta/b$ plane within $(\gamma_2,\theta)\in(0,1)\times(0,L)$, as shown in Fig. \ref{fig9}(b)-(c). Based on Proposition \ref{propB.1}, we get the following results with tedious but elementary calculations.

\begin{prop}\label{propB.2}
    Let $b>0$, $\gamma_2\in(0,1)$ and $\theta\in(0,L)$. Then the function $G(s)$ defined by \eqref{B3} has exactly one real root in $(0,+\infty)$.
\end{prop}

We can now prove the main result of this appendix.
\begin{lem}\label{lemB.3}
    The rescaled system \eqref{eq4.8} with \eqref{eg2} has a unique coexistence equilibrium $Q_*=( u_*, v_*, w_*)$ if and only if $\theta\in(\Theta_1, L)$. Moreover,
    \begin{align}\nonumber
         u_* =\frac{\gamma_2  v_*  (b-\theta )(v_M- v_*)}{\gamma_2  v_*  (2 b-\theta )+b-\theta },\quad
             w_*=(1- v_* ) \xkh{ v_* +\frac1\gamma_2},
    \end{align}
    and $v_*\in(0,\min\dkh{1,v_M})$ satisfies $G(v_*)=0$, where the positive constant $v_M$ and the function $G$ are given by \eqref{B1} and \eqref{B3}, respectively.
\end{lem}

\begin{proof}
    Clearly, it follows from \eqref{B3} and Proposition \ref{propB.1} that $G(0)=D_0<0$, which alongside Proposition \ref{propB.2} implies that the rescaled system \eqref{eq4.8} with \eqref{eg2} has a unique coexistence equilibrium $Q_*=( u_*, v_*, w_*)$ if and only if $G(\min\dkh{1,v_M}) >0$. If $\theta\in[L_2,L)$, then \eqref{B2} implies $G(\min\dkh{1,v_M})=G(1)=2b \gamma_2 (\gamma_2+1)^2 (L-\theta )>0$. If $\theta\in(0,L_2)$, then
    \begin{align}\nonumber
        G(\min\dkh{1,v_M})= G(v_M)
        = \frac{b^4 \varphi_2(\theta)}{\gamma_2 (b-\theta )^4},
    \end{align}
    where $\varphi_2(\theta)$ is given by \eqref{addeq1}, and $\varphi_2(\theta)>0$  if and only if $\theta>\Theta_1$. The proof is completed.
\end{proof}


\end{document}